\documentclass[runningheads,anonymous]{llncs}
\usepackage{graphicx}
\usepackage{hyperref}

\usepackage[subtle]{savetrees}
\usepackage{amssymb}
\usepackage{lmodern}        
\usepackage[T1]{fontenc}  
\usepackage[utf8]{inputenc} 

\usepackage{placeins} 
\usepackage{booktabs} 
\usepackage{multirow} 
\usepackage{xcolor}

\usepackage{paralist} 
\usepackage{verbatim} 

\usepackage{array} 
\usepackage{mathtools} 
\usepackage{amssymb} 
\usepackage{algorithm}
\usepackage{thmtools,thm-restate}
\usepackage[noend]{algpseudocode}
\usepackage{subcaption}
\usepackage[pdf]{graphviz}

\newcommand{\Alltrans}{Q_E\times Act_E\times Q_E}
\newcommand{\para}[1]{\noindent\textbf{#1}}

\newif\iflong

\begin{document}
\title{Safe Environmental Envelopes of Discrete Systems
 }
\titlerunning{Safe Environmental Envelopes of Discrete Systems}
%
\author{}
\institute{}
\author{R\^omulo Meira-G\'oes\inst{1} 
\and
Ian Dardik\inst{2} \and
Eunsuk Kang\inst{2} \and
St\'ephane Lafortune\inst{3} \and\\
Stavros Tripakis\inst{4}
}
\authorrunning{R. Meira-G\'oes et al.}
\institute{
School of EECS,\\ Pennsylvania State University, State College USA\\ \email{romulo@psu.edu}\and 
School of Computer Science,\\ Carnegie Mellon University, Pittsburgh USA\\ \email{\{idardik,eunsukk\}@andrew.cmu.edu}\and
EECS Department, University of Michigan, Ann Arbor USA\\
\email{stephane@umich.edu}\\
\and
Khoury College of Computer Science, Northeastern University, Boston USA\\
\email{stavros@northeastern.edu}}
\maketitle              

\begin{abstract}
A safety verification task involves verifying a system against a desired safety property under certain assumptions about the environment.
However, these environmental assumptions may occasionally be violated due to modeling errors or faults.
Ideally, the system guarantees its critical  properties even under some of these violations, i.e., the system is \emph{robust} against environmental deviations.
This paper proposes a notion of \emph{robustness} as an explicit, first-class property of a transition system that captures how robust it is against possible \emph{deviations} in the environment.
We modeled deviations as a set of \emph{transitions} that may be added to the original environment.
Our robustness notion then describes the safety envelope of this system, i.e., it captures all sets of extra environment transitions for which the system still guarantees a desired property.
We show that being able to explicitly reason about robustness enables new types of system analysis and design tasks beyond the common verification problem stated above. 
We demonstrate the application of our framework on case studies involving a radiation therapy interface, an electronic voting machine, a  fare collection protocol, and a medical pump device.
\end{abstract}

\graphicspath{{Figs/}}

\section{Introduction}

A common type of verification task involves verifying a system ($C$) against a desired property ($P$) under certain assumptions about the environment ($E$); i.e., $C || E \models P$. 
Such assumptions may capture, for example, the expected behavior of a human operator in a safety-critical system,  the reliability of the communication channel in a distributed system, or the capabilities of an attacker. However, the actual environment ($E'$) may occasionally deviate from the original model ($E$), due to changes or faults in the environment entities (e.g., errors committed by the operator or message loss in the channel). 
For certain types of deviations, a system that is \emph{robust} would ideally be able to guarantee the property even under the deviated environment; i.e., $C || E' \models P$.

This paper proposes the notion of \emph{robustness} as an explicit, first-class property of a transition system that captures how robust it is against possible \emph{deviations} in the environment. 
A deviation is modeled as a set of \emph{extra transitions} that may be added to the original environment, resulting in a new, deviated environment $E'$ that has a larger  set of behaviors than $E$ does. 
Then, system $C$ is said to be \emph{robust} to this deviated environment with respect to $P$ if and only if it can still guarantee $P$ even in presence of the deviation.
Finally, the overall \emph{robustness} of $C$ with respect to $E$ and $P$, denoted $\Delta$, is the largest set of   deviations that the system is robust against.

Conceptually, $\Delta$ defines the safe operating envelopes of the system: As long as the deployment environment remains within these envelopes, the system can guarantee a desired property. 
Being able to explicitly reason about $\Delta$ enables new types of system analysis and design tasks beyond the common verification problem stated above. 
Given a pair of alternative system designs, $C_1$ and $C_2$, one could rigorously compare them with respect to their robustness levels; they both may satisfy property $P$ under the normal operating environment $E$, but one may be more robust to deviations than the other. 
Given two properties, $P_1$ and $P_2$ (the latter possibly more critical than the former), one could check whether the system would continue to guarantee $P_2$ under a deviated environment even if it fails to do so for $P_1$. 
Finally, given $E$, $P$, and a desired level of robustness, $\Delta$, one could \emph{synthesize} machine $C$ to be robust to $\Delta$.

In this paper, we formalize (1) the proposed notion of robustness and (2) the problem of computing $\Delta$ for given $C$, $E$, and $P$.
One approach to automatically compute $\Delta$ is a brute-force method that enumerates all possible sets of deviations; however, as we will show, this approach is impractical, as the number of deviations is exponential in the size of the environment. 
To mitigate this, we present an approach for computing $\Delta$ by reduction to a controller synthesis problem~\cite{Pnueli:1989a,Ramadge:1987}.

We have built a prototype of the proposed approach for computing robustness and applied it to several case studies, including models of (1) a radiation therapy interface, (2) an electronic voting machine, (3) a public transportation fare collection protocol, and (4) a medical pump device. 
Our results show that our approach is capable of computing $\Delta$ to provide information about deviations under which these systems are able to guarantee their critical safety properties.

The contributions of this paper are as follows: (i) A novel, formal definition of robustness against environmental deviations (Sect.~\ref{sect:robust-notion}); (ii) A simple, brute-force method for computing robustness and a more efficient approach based on controller synthesis (Sect.~\ref{sect:computing-robustness}); and
(iii) A prototype tool for computing $\Delta$ and an experimental evaluation on several case studies (Sect.~\ref{sect:experiments}).

\section{Motivating example}\label{sect:motivating}

As a motivating example, we consider the Therac-25 radiation therapy machine.
This machine is infamous for a design flaw  that caused radiation overdoses, several of which led to the deaths of patients who received treatment \cite{Leveson:1993}.
In this section, we introduce a model for the Therac-25 based on the descriptions in \cite{Leveson:1993} and discuss several methods for analyzing its safety.
We show that robustness provides a generally richer analysis than classic verification.

\subsubsection{System }
We model the Therac-25 as the composition of the following three finite-state machines:
(1) $C_{term}$, a computer terminal that nurses use to operate the Therac-25,
(2) $C_{beam}$, a beam-emitter that fires a radiation treatment beam in either \textit{X-ray} or \textit{electron} mode, and
(3) $C_{turn}$, a turntable that rotates between two hardware components called the \textit{flattener} and the \textit{spreader}.
Formally, we define the Therac-25 as the composition all three machines:  $C_{T25} = C_{term}||C_{beam}||C_{turn}$.
We show the terminal and turntable in Figs. \ref{fig:therac-terminal} and \ref{fig:therac-turntable} respectively.
We show the beam in Sect. \ref{sect:therac} (Fig. \ref{fig:therac-beam}), where we present a case study on the Therac-25.

\para{Environment }
Nurses operate the Therac-25 by typing at a keyboard connected to a terminal. 
A nurse begins by choosing a beam mode by typing either an ``x'' for X-ray or an ``e'' for electron mode.
The nurse then hits the ``enter'' key and waits for the terminal to display ``beam ready'' before finally pressing the ``b'' key to fire the beam.
This workflow defines the operating environment which we call $E$, shown in Fig.~\ref{fig:therac-env}.

\para{Safety property }
Since the X-ray beams contain a high concentration of radiation, it is imperative that the flattener is in place when the machine fires an X-ray.
We capture this key safety property in the following LTL \cite{Pnueli:1977} formula:
$$\textbf{G}\big(\text{XFIRED}\ \rightarrow \ \text{FLATMODE}\big)$$
In this formula, XFIRED is a predicate that is true if an X-ray beam was just fired, while FLATMODE is a predicate that is true when the turn table is in flattener mode.
We refer to this safety property as $P_{xflat}$ in this example.

\begin{figure}[!t]
\begin{subfigure}[b]{0.48\textwidth}
    \centering
    \includegraphics[height=75px]{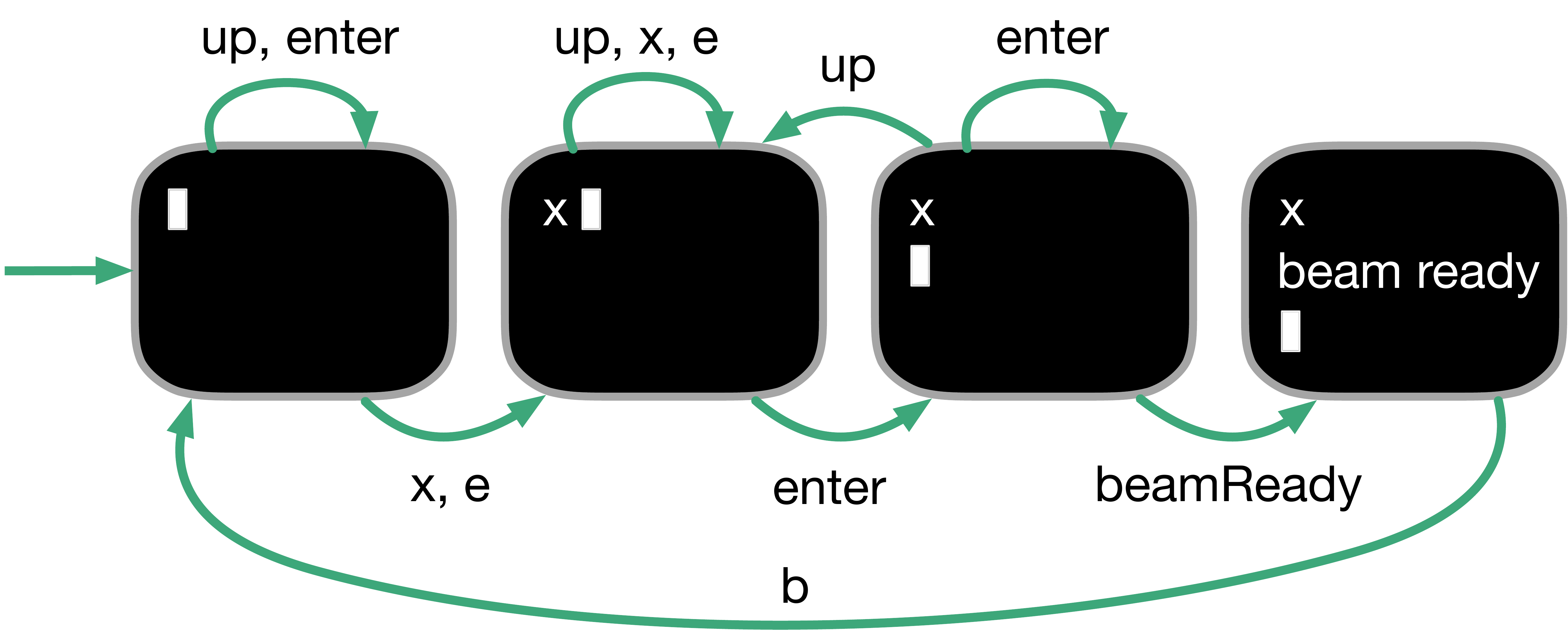}
    \caption{The operating terminal, $C_{term}$.}
    \label{fig:therac-terminal}
\end{subfigure}
\hspace{25px}
\enskip
\begin{subfigure}[b]{0.48\textwidth}
    \centering
    \includegraphics[height=100px]{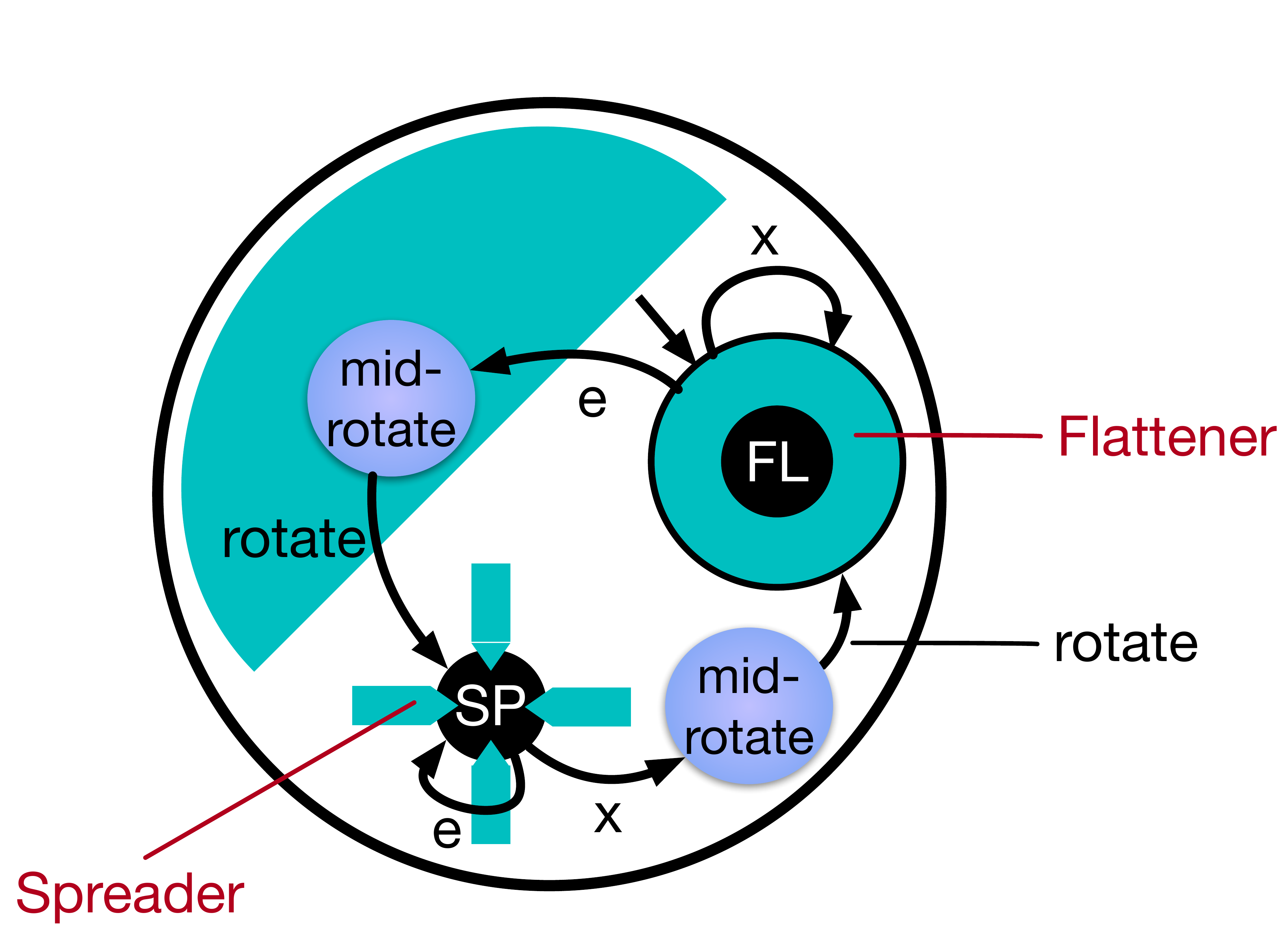}
    \caption{The turntable, $C_{turn}$.}
    \label{fig:therac-turntable}
\end{subfigure}
\enskip
\begin{subfigure}[b]{\textwidth}
    \vspace{10px}
    \centering
    \includegraphics[height=30px]{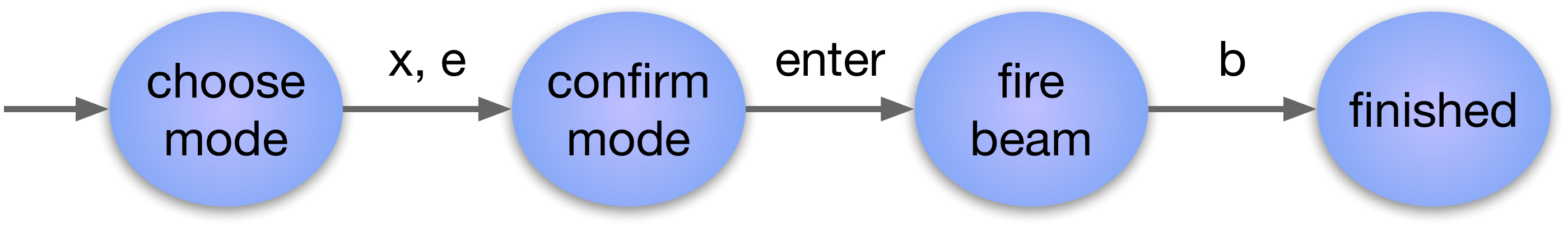}
    \caption{The normative environment, $E$}
    \label{fig:therac-env}
\end{subfigure}
\caption{The Therac-25 is modeled as $C_{T25}=C_{term}||C_{beam}||C_{turn}$.  $C_{beam}$ is in Figure \ref{fig:therac-beam}.
}
\label{fig:therac-term-turn}
\end{figure}
\para{Safety Analyses }
Robustness opens our safety analysis beyond classic verification.
We discuss several analysis options below.

\para{(1) Standard Verification:}
We can check that the Therac-25 is safe within the operating environment, that is, $E||C_{T25} \models P_{xflat}$.
Standard model checking techniques \cite{Baier:2008} show that the Therac-25 is indeed safe with respect to $E$.

\para{(2) Robustness Calculation:}
Given that the Therac-25 is safe with respect to $E$, we can calculate its robustness $\Delta$.
This calculation identifies the set of safe environmental envelopes of the Therac-25.
Importantly, these envelopes reveal the environmental deviations that the Therac-25 can safely handle.
For example, in Sect.~\ref{sect:therac}, we show that the Therac-25 is robust against the environmental deviations in Fig. \ref{fig:therac-tol} in which a nurse repeatedly hits ``enter'' or the ``up'' arrow key after choosing a beam mode.

\para{(3) Controller Comparison:}
Holding the environment $E$ and the property $P_{xflat}$ constant, we can compare the robustness of the Therac-25 against other models.
In Sect.~\ref{sect:therac}, we introduce the Therac-20 ($C_{T20}$) and compare the robustness between $C_{T25}$ and $C_{T20}$.
Although both machines are safe with respect to the normative environment, we will find that $C_{T25}$ is strictly less robust than $C_{T20}$.
We will show how contrasting the robustness between the two machines exposes a critical software bug in the Therac-25.
Furthermore, we will show that fixing the bug in the Therac-25 causes its robustness to be equivalent to the Therac-20.

\para{(4) Property Comparison:}
Holding the environment $E$ and the machine $C_{T25}$ constant, we can compare the machine's robustness with respect to $P_{xflat}$ and a second safety property.
For example, we could consider a new safety property $P'$ that strengthens $P_{xflat}$ by additionally enforcing the spreader to be in place when a beam is fired in electron mode.
The property $P'$ might be of interest to avoid an \textit{underdose}, a situation that might result from the flattener being in place when an electron beam is fired.
Because $P'$ is stronger than $P_{xflat}$, a designer may be interested to compare the robustness between the properties to understand which environmental deviations maintain $P_{xflat}$, but violate $P'$.
\section{Modeling formalism}\label{sect:preliminaries}
This section describes the underlying formalism used to model the environment, controlled systems, and the properties enforced by them.

\subsubsection*{Labeled transition systems}
Given a finite set $A$, the usual notations $|A|$ and $A^*$ denote the cardinality of $A$ and  the set of all finite sequences over $A$ respectively.
In this work, we use finite labeled transition systems to model the behavior of the environment, the controller, and the property.
\begin{definition} 
A \emph{labeled transition system} (LTS) $E$ is a tuple $\langle Q_E, Act_E, R_E, q_{0,E}\rangle$, where $Q_E$ is a finite set of states, $Act_E$ is a finite set of actions, $R_E\subseteq Q_E\times Act_E\times Q_E$ is the transition relation of $E$, and $q_{0,E}\in Q_E$ is the initial state.
\end{definition}
LTS $E$ is said to be deterministic if for any $(q,a,q'),(q,a,q'')\in R_E$, then $q'=q''$; otherwise it is nondeterministic.
We extend the transition relation $R_E$ to finite sequences of actions as ${R_E}^*\subseteq Q_E\times {Act_E}^*\times Q_E$ in the usual manner.
A \emph{trace} of $E$ is a finite sequence of actions $a_0\dots a_n$ of $E$ complying with the transition in ${R_E}^*$, i.e., $(q_{0,E}, a_0\dots a_n, q)\in {R_E}^*$ for some $q\in Q_E$.
The set of all traces in $E$ is denoted by $beh(E)$.

Given LTSs $E_1$ and $E_2$, the parallel composition $||$ defines standard synchronization of $E_1$ and $E_2$ \cite{Baier:2008,Lafortune:2021}.
The composed LTS $E_1||E_2 = \langle Q_{E_1}\times Q_{E_2}, Act_{E_1}\cup Act_{E_2}, R_{E_1||E_2},\allowbreak (q_{0,E_1},q_{0,E_2})\rangle$ synchronizes over the common actions between $E_1$ and $E_2$ and interleaves the remaining actions.
Lastly, given LTSs $E_1$ and $E_2$, we say that $E_1$ is a subset of $E_2$, denoted $E_1\subseteq E_2$, if $Q_{E_1}\subseteq Q_{E_2}$, $Act_{E_1}= Act_{E_2}$, $R_{E_1}\subseteq R_{E_2}$, and $q_{0,E_1}= q_{0,E_2}$.

\subsubsection*{Control strategy}
Let an LTS $E$ represent the environmental model to be controlled.
A control strategy, or simply \emph{controller}, for $E$ is a function that maps a finite sequence of actions to a set of actions, i.e., $C:{Act_E}^{*}\rightarrow 2^{Act_E}$.
A \emph{controlled trace} of $E$ is a trace of $E$, $a_0\dots a_n\in beh(E)$, such that $a_{i} \in C(a_0\dots a_{i-1})$ for any $i\leq n$.
The set of all controlled runs, denoted by $beh(E/C)$, defines the closed-loop system of $C$ controlling $E$.
For convenience, this closed-loop system is denoted by $E/C$.
In this work, we assume that controller $C$ has finite memory and it can be represented by a deterministic LTS.
With an abuse of notation, the LTS controller representation is also denoted by $C$.
For convenience, we define controller $C = \langle Q_C,Act_C,R_C,q_{0,C}\rangle$ to have the same actions as in $E$, i.e., $Act_C = Act_E$.
In this manner, the closed-loop system $E/C$ can be represented by the composition of environment $E$ and controller $C$: $E/C = E||C$.

\begin{remark}
We assume that all elements of the set of actions $Act_E$ are ``controllable'' actions, that can be acted upon by a controller.
However, the nondeterministic transition relation of $E$ can be used to model uncontrollable actions of the environment.
After an action $a$ is selected by the controller at state $q$, the environment decides which state the system will be in, similarly to two-player games \cite{Gradel:2002}.
\end{remark}


\subsubsection*{Safety property}
In this work, we consider a  class of regular linear-time properties called safety properties over an environment $E$ \cite{Baier:2008}.
A safety property $P$ is represented by a deterministic LTS $P$ that defines the set of accepted behaviors.
Usually, the LTS $P$ encodes both the traces that satisfy $P$ and those that violate it by including a sink error state.
Formally, any trace that reaches the error state $err\in Q_P$ violates the safety property.
An LTS $E$ satisfies property $P$, denoted by $E\models P$, whenever the traces in $beh(E)$ do not reach the error state in $P$.
In this manner, we can test if $E\models P$ by composing $E||P$ and investigating if the $err$ is reached.

\begin{example}\label{ex:running}
We describe a simple example that we use as a running example throughout the paper. 
Figure~\ref{fig:ex-algo-constr} depicts the environment $E$, controller $C$, and property $P$ considered in this example.
The environment $E$ defines that action $a$ is immediately followed by action $b$.
Although controller $C$ in Fig.~\ref{fig:ctr-alg-lts} only shows action $a$, we assume that $Act_C = \{a,b\}$.
In this manner, $C$ only allows action $a$ to occur.
Lastly, property $P$ defines that action $a$ should happen at most two times while action $b$ should never happen.
It follows that $E/C\models P$ since the controller disables action $b$ and the environment only executes one instance of action $a$.
\end{example}
\begin{figure}[!t]
\centering  
\begin{subfigure}[b]{0.35\textwidth}
        \centering
        \includegraphics[width=.5\textwidth]{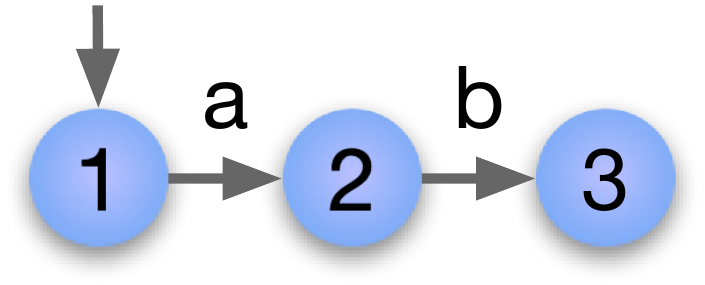}
        \caption{Environment $E$}
        \label{fig:env-alg-lts}
\end{subfigure}
\enskip
\begin{subfigure}[b]{0.2\textwidth}
        \centering
        \includegraphics[width=.41\textwidth]{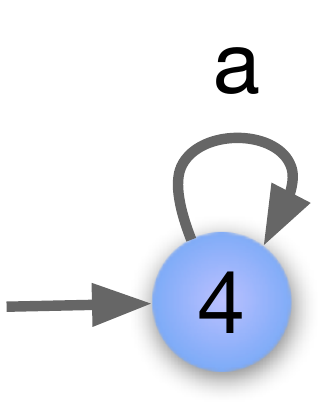}
        \caption{Controller $C$}
        \label{fig:ctr-alg-lts}
\end{subfigure}
\begin{subfigure}[b]{0.35\textwidth}
        \centering
        \includegraphics[width=.81\textwidth]{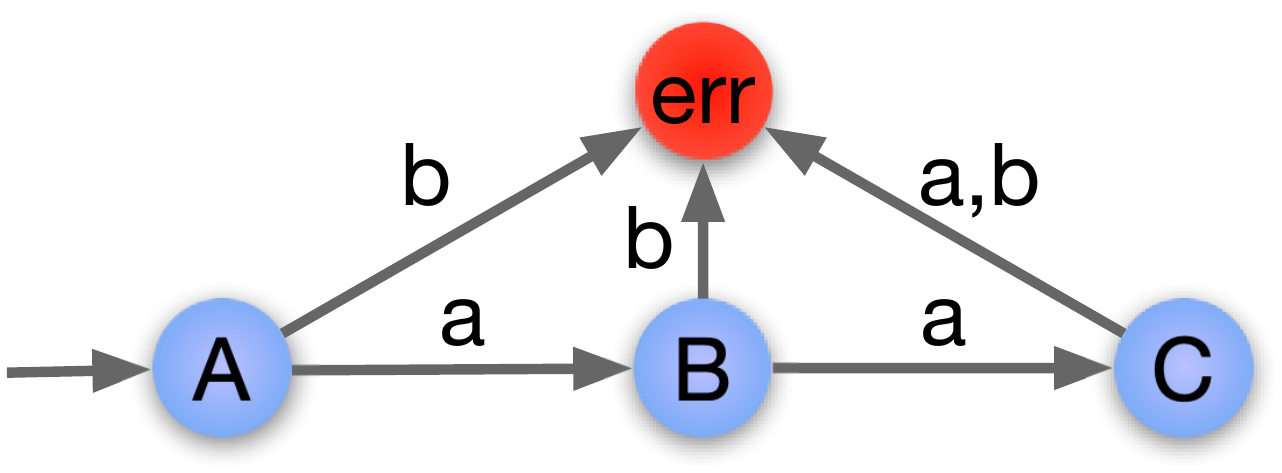}
        \caption{Property $P$}
        \label{fig:prop-alg-lts}
\end{subfigure}
\caption{LTSs for the running example}
\label{fig:ex-algo-constr}
\end{figure}

\section{Robustness against environmental deviations} \label{sect:robust-notion}
\subsection{Deviations}
\emph{A deviation} is a set of transitions $d\subseteq (Q_E\times Act_E\times Q_E)$
A \emph{deviated system} is defined by augmenting the transitions of environment $E$ with a deviation set:  
\begin{definition}
Given an LTS $E = \langle Q_E, Act_E, R_E, q_{0,E}\rangle$ and a deviation $d \subseteq Q_E\times Act_E \times Q_E$.
We define the \emph{deviated system} $E_{d}$ as $E_{d}:=\langle Q_E, Act_E, R_E\cup d, q_{0,E}\rangle$.
\end{definition}

A controller $C$ that guarantees property $P$ for environment $E$, i.e., $E/C\models P$, might violate this property for the deviated environment $E_d$, i.e., $E_d/C\not\models P$. 
\begin{definition}\label{def:robust-dev}
Controller $C$ is a \emph{robust controller} with respect to environment $E$, deviation $d$, and property $P$ if $E_{d}/C\models P$.
Deviation $d$ is a \emph{robust deviation} with respect to $E$, $C$, and $P$ if $C$ is a robust controller with respect to $E$, $d$, and $P$.
\end{definition}
\begin{remark}
In this paper, we are only interested in ensuring safety properties over the controlled system.
For this reason, it is sufficient to only consider adding new transitions to the environment.
If a controlled system is safe, then deleting transitions from the environment does not violate the safety property. 
\end{remark}

\subsection{Comparing deviations}\label{sect:comparing-perts}
Each deviation set affects the environment in different ways.
To reason about the effects of each deviation set, we compare them using a partial order relation over $Q_E\times Act_E\times Q_E$.
For deviations $d_1$ and $d_2$ such that $d_1\subseteq d_2$, $d_2$ deviates LTS $E$ more than $d_1$ since $beh(E_{d_1})\subseteq beh(E_{d_2})$.
For this reason, we select the relation $\subseteq$ over $Q_E\times Act_E\times Q_E$ to be the partial order to compare different deviation sets.

\begin{definition}\label{def:dev-order}
\hspace*{-.1cm}Given $E$ and deviations $d_1, d_2$, $d_1$ is \emph{at least as powerful} as $d_2$ if $d_2\subseteq d_1$.
\end{definition}
%

\subsection{Robustness}
Intuitively, robustness is defined as the set of all possible robust deviations $d$ with respect to the environment $E$, controller $C$, and safety property $P_{saf}$.
Additionally, we introduce an environmental constraint, $P_{env}$, to capture domain knowledge about the system under analysis.
$P_{env}$ will filter environment deviations that might not be physically feasible or of interest to analyze.
This constraint is captured as a safety property over $E$, i.e., $E\models P_{env}$ states that the environment satisfies the constraint.
Formally, our robustness notions is defined as follows:
\begin{definition} \label{def:robustness}
Let environment $E$, controller $C$, property $P_{saf}$ such that $E/C\models P_{saf}$, and environment constraint $P_{env}$ such that $E\models P_{env}$ be given.
The robustness of controller $C$ with respect to $E$, $P_{saf}$, and $P_{env}$, denoted by $\Delta(E,C,P_{saf},P_{env})$, is a set of robust deviations $\Delta\subseteq 2^{\Alltrans}$.
$\Delta$ is defined to be the (unique) set of robust deviations satisfying the following conditions:
\begin{enumerate}
\item $\forall d \in \Delta.\ E_{d}/C\models P_{saf}$  \emph{[$d$ is robust]};
\item $\forall d\subseteq \Alltrans .E_{d}/C\models P_{saf}\wedge E_{d}\models P_{env}\Rightarrow\exists d'\in \Delta.d\subseteq d'$ \emph{[$d$ is represented]}; 
\item $\forall d,d'\in \Delta.\ d\neq d'\Rightarrow d \not\subseteq d'$ \emph{[unique representation]}.
\item $\forall d\in \Delta.\ E_d\models P_{env}$ \emph{[d is feasible]}.
\end{enumerate}
When $E,C, P_{saf}$, and $P_{env}$ are clear from context, we simply write $\Delta$.
The set $\Delta$ is also denoted as the safety envelope of $C$ with respect to $E$, $P_{saf}$, and $P_{env}$.
\end{definition}
Intuitively, the set $\Delta$ defines an upper bound on the possible deviations from $E$ that controller $C$ is robust against.
In other words, $\Delta$ captures the envelopes for which controller $C$ remains safe.

If a designer does not have domain knowledge about the system, then $P_{env}$ can be set to not constrain the environment, i.e., $P_{env} = Act_E^*$.
After computing $\Delta$ without environmental constraints, a designer can obtain important information about the system and the environment.
In the next analysis iteration, this knowledge can be transformed into environmental constraints to enhance the robustness analysis, i.e., $P_{env} \subseteq Act_E^*$. 

By definition, $\Delta$ is always non-empty since $d = \emptyset$ is always robust.
Moreover, due to conditions 2 and 3, only maximal robust deviations are included in $\Delta$.
We show that there is a unique set of deviations that satisfies the conditions of Def.~\ref{def:robustness}. The proof of this lemma is in the appendix. 
\begin{restatable}{lemma}{lemmarobustnessuniqueness}\label{lemma:robustness_uniqueness}
Given LTS $E$, controller $C$, safety property $P_{saf}$, and environment property $P_{env}$, there is a unique $\Delta$ that satisfies the conditions in Def.~\ref{def:robustness}. 
\end{restatable}

\begin{example}\label{example:robustness}
Back to our running example, we investigate robust deviations and $\Delta$.
For simplicity, we do not impose any environment constraint, i.e., $P_{env} = Act_E^*$.
Figure~\ref{fig:robust-envs} shows four robust deviations for our running example, where transitions in green are deviations added to the environment.
All robust deviations allow at most two transitions with action $a$, which is the maximum number allowed by the property.
In this example, $\Delta$ has three robust deviations that are represented in Figs.~\ref{fig:tol-safe2-meta}-\ref{fig:tol-safe4-meta}.
Since the robust deviation shown in Fig.~\ref{fig:tol-safe-meta} is a subset of both deviations in Fig.~\ref{fig:tol-safe2-meta} and Fig.\ref{fig:tol-safe3-meta}, it is not included in $\Delta$.

\begin{figure}[!t]
\centering  
\begin{subfigure}[b]{0.48\textwidth}
        \centering
        \includegraphics[width=.5\textwidth]{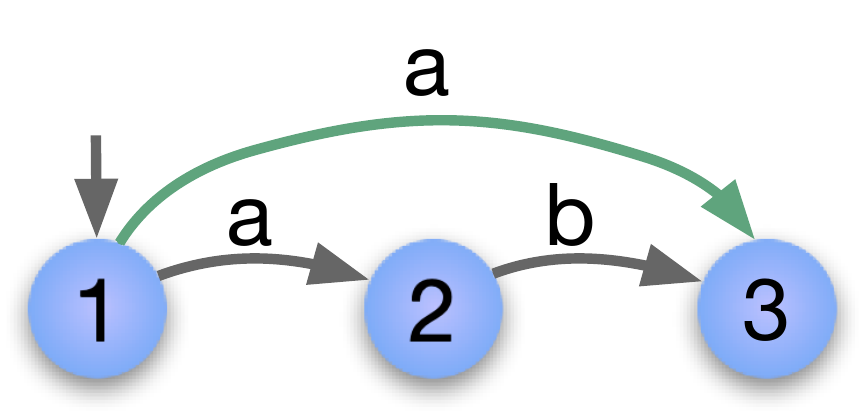}
        \caption{A robust deviated environment}
        \label{fig:tol-safe-meta}
\end{subfigure}
\enskip
\begin{subfigure}[b]{0.49\textwidth}
        \centering
        \includegraphics[width=.5\textwidth]{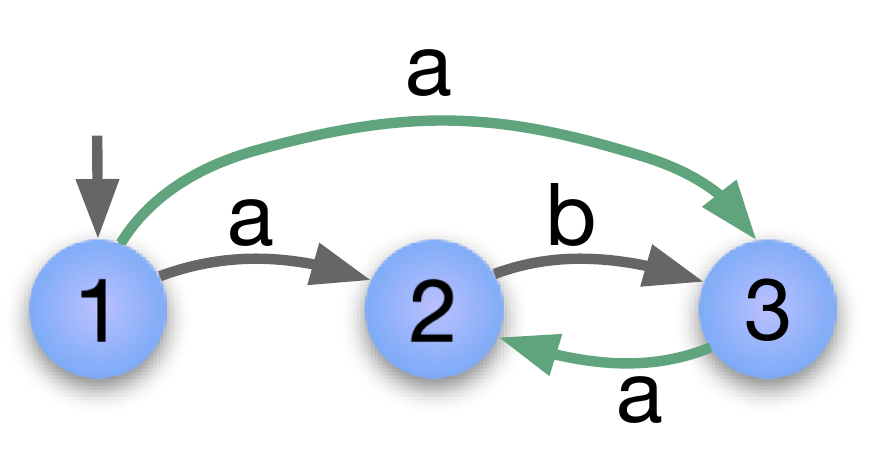}
        \caption{Maximal robust deviated environment}
        \label{fig:tol-safe2-meta}
\end{subfigure}
\begin{subfigure}[b]{0.49\textwidth}
        \centering
        \includegraphics[width=.5\textwidth]{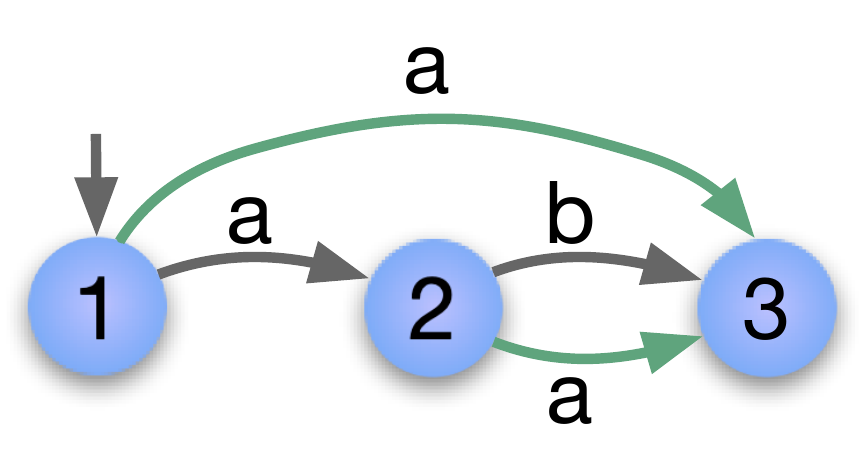}
        \caption{Maximal robust deviated environment}
        \label{fig:tol-safe3-meta}
\end{subfigure}
\begin{subfigure}[b]{0.49\textwidth}
        \centering
        \includegraphics[width=.5\textwidth]{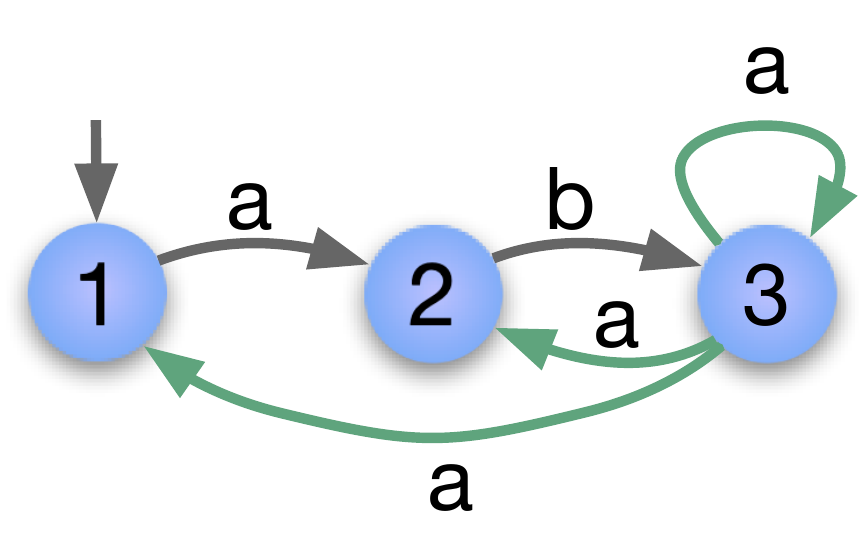}
        \caption{Maximal robust deviated environment}
        \label{fig:tol-safe4-meta}
\end{subfigure}

\caption{Robust deviated environments. Robust transitions $Q_E\times\{b\}\times Q_E$ are omitted.}
\label{fig:robust-envs}
\end{figure}
\end{example}

\subsection{Problem statement}
Although Def.~\ref{def:robustness} has formally introduced our notion of robustness, it does not show how to compute robustness.
Therefore, we investigate the problem of computing the set $\Delta$.

\begin{problem}\label{prob:comp-tol}
Given $E$, $C$, $P_{saf}$, and $P_{env}$ as in Def.~\ref{def:robustness}, compute $\Delta$.
\end{problem}

\subsection{Comparing robustness}\label{sect:compare-robustness}
Our robustness definition also allows us to compare the robustness between different controllers as well as different safety properties.

\subsubsection{Comparing controllers}
Holding the environment and safety property constant, we can compare the robustness of the controllers.

\begin{definition} \label{def:compare-controllers}
Given an environment $E$, controllers $C_1$ and $C_2$, safety property $P_{saf}$, and environment constraint $P_{env}$, controller $C_1$ is at least as robust as $C_2$ if and only if for all $d_2 \in \Delta(E,C_2,P_{saf},P_{env})$ there exists $d_1 \in \Delta(E,C_1,P_{saf},P_{env})$ such that $d_2 \subseteq d_1$.
Equality and strictly less/more robust are defined in the usual manner using $\subseteq$.
\end{definition}

\subsubsection{Comparing safety properties}
Holding the environment and controller constant, we can compare the robustness between safety properties.

\begin{definition} \label{def:compare-properties}
Given an environment $E$, controllers $C$, safety properties $P_{saf,1}$ and $P_{saf,2}$, and environment constraint $P_{env}$, controller $C$ is at least as robust with respect to $P_{saf,1}$ than with respect to $P_{saf,2}$ if and only if for all $d_2 \in \Delta(E,C,P_{saf,2},P_{env})$, there exists $ d_1 \in \Delta(E,C,P_{saf,1},P_{env})$ such that $ d_2 \subseteq d_1$.
\end{definition}



\section{Computing robustness}
\label{sect:computing-robustness}

This section presents two manners of solving Problem~\ref{prob:comp-tol}.
One is a brute-force algorithm whereas the second uses control techniques to obtain the solution.
Usually when dealing with regular safety properties, one transforms the safety property into an invariance property.
This transformation is simply obtained by composing the environment with the safety property; then, an invariance property equivalent to the safety is defined over this composed system \cite{Baier:2008}.
In this composed system, an invariance property is simply defined by a set of safe states. 
Unfortunately, computing robustness for safety properties does not directly reduce to computing robustness for invariance properties.

When transforming a safety property $P_{saf}$ to an invariance property, we compose the environment and the safety property.
Let us assume that there are no environmental constraints.
In our scenario, the invariance property $P_{inv}$ is defined based on the composed system $E||C||P_{saf}$, i.e., $P_{inv}\subseteq Q_{E||C||P_{saf}}$.
The composed system $P_{inv}$ introduces memory to the environment to differentiate when the safety property is violated or not.
This memory addition prevents a simple reduction between invariance and safety properties since robustness is defined with respect to the environment.
Robustness defines new transitions in $E$ whereas computing robustness with respect to $P_{inv}$ defines new transitions in $E||C||P_{saf}$.
For this reason, we cannot simply reduce the problem of computing $\Delta$ with respect to safety properties to the problem of computing $\Delta$ with respect to an invariance property.



\subsection{Brute-force algorithm}\label{sect:brute-force}
One way of solving Prob.~\ref{prob:comp-tol} is via a brute-force algorithm. 
Intuitively, this algorithm is broken into two parts: (i) finding the set of robust deviations that satisfy the environmental constraint, and (ii) identifying the maximal ones within this set.
In part (i), we verify $E_d||C\models P_{saf}$ and $E_d\models P_{env}$ for all deviations $d\subseteq (Q_E\times Act_E\times Q_E)\setminus R_E$, which can be solved using standard model checking techniques \cite{Baier:2008}.
Since this algorithm checks if every deviation set is robust or not, it is clear that it computes $\Delta$.

\subsection{Controlling the deviations without environmental constraints}
\label{subsect:unconstraint}

Due to the lack of scalability of the brute-force algorithm, we search for more efficient ways to compute $\Delta$.
For readability purposes, we start by describing our algorithm in detail assuming no environmental constraints, i.e., unconstrained environment $P_{env} = Act_E^*$.
In the next section, we show how to use this algorithm to completely solve Prob.~\ref{prob:comp-tol}, i.e., for a possibly constrained environment $P_{env} \subseteq Act_E^*$.


\subsubsection{Overview of the control algorithm}
At a high level, we transform the problem of computing $\Delta$ to a problem of controlling environmental transitions to avoid safety violations.
Intuitively, we control deviations to force them to be robust, i.e., we take the viewpoint that we can control transitions in $(\Alltrans) \setminus R_E$.
Different ways of controlling transitions in $(\Alltrans) \setminus R_E$ provide different robust deviations.

\begin{figure}[thpb]
\centering  
\includegraphics[width=\textwidth]{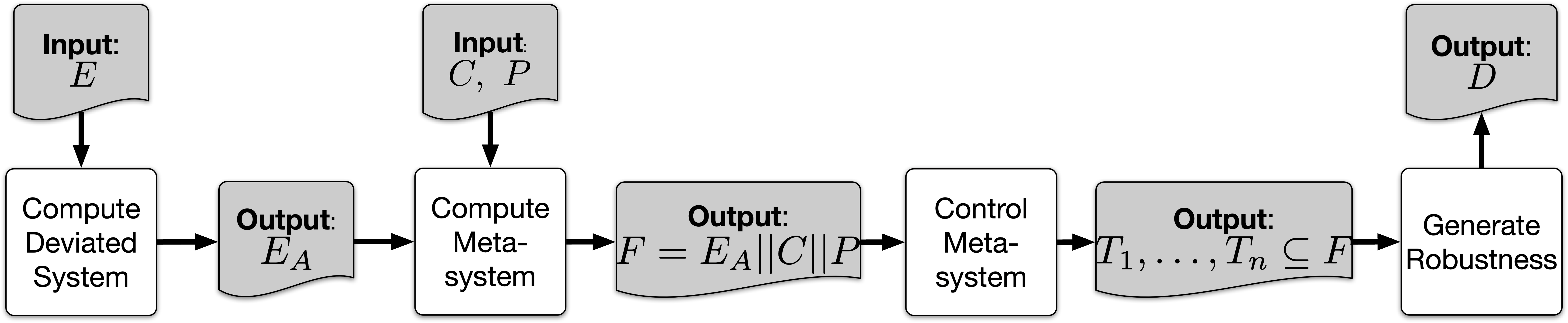}
\caption{Overview of our approach to compute robustness for the unconstrained environment. The inputs are the LTSs of environment $E$, controller $C$, and property $P_{saf}$. The set $A$ is the set of all environment transitions, $A = \Alltrans$. The LTSs $T_1,\dots, T_n\subseteq F$ represent controlled meta-systems.}
\label{fig:overview-alg}
\end{figure}

Figure~\ref{fig:overview-alg} provides an overview of our approach.
First, we define LTS $E_A$ to be the deviated system with all possible transitions, i.e., $A = \Alltrans$.
The deviated system $E_A$ is the maximally deviated environment since it encompasses every possible deviated system $E_d$ for $d\subseteq \Alltrans$.

Next, we compose the deviated environment $E_A$ with controller $C$ and property $P_{saf}$, to create a ``meta-system'' $F$.
This meta-system provides information about how the deviated environment $E_A$ under the control of $C$ can violate $P_{saf}$.
Following this composition, we pose a control problem over the meta-system to prevent any violation of $P_{saf}$.
There are multiple ways of controlling this composed system; in our approach, we obtain a finite number of controllers encoded as $T_i\subseteq F$.
These different ways of controlling the meta-system provide different robust deviations from which we can extract $\Delta$.
To make our approach concrete, we describe each step in detail using our running example, shown in Figure~\ref{fig:ex-algo-constr}.


\subsubsection{Constructing the meta-system}

The deviated environment $E_A = E_{\Alltrans}$ contains the behavior of any other deviated environment.
Therefore, we define the meta-system to be the composition of deviated environment $E_A$, controller $C$, and property $P_{saf}$, i.e., $F = E_A||C||P_{saf}$.
Figure~\ref{fig:meta-system} shows the meta-system $F$ for our running example.
Since $C$ only has one state, we omit its state from the state names in Fig.~\ref{fig:meta-system}, i.e., states in Fig.~\ref{fig:meta-system} are defined as $(q_e,q_p)\in Q_E\times Q_{P_{saf}}$ instead of $(q_e,q_c,q_p)\in Q_E\times Q_C \times Q_{P_{saf}}$.
All transitions in $F$ are labeled $a$, omitted in Fig.~\ref{fig:meta-system}, since controller $C$ only enables action $a$.
We also identify in $F$ which transitions are derived from the environment (dashed blue) and which are derived from deviations (green).
For simplicity, we define a single error state in $F$ to capture every $(q_e,q_c,err)\in Q_E\times Q_C\times Q_{P_{saf}}$.

\begin{figure}[thpb]
\centering  
\begin{subfigure}[b]{0.44\textwidth}
        \centering
        \includegraphics[width=\textwidth]{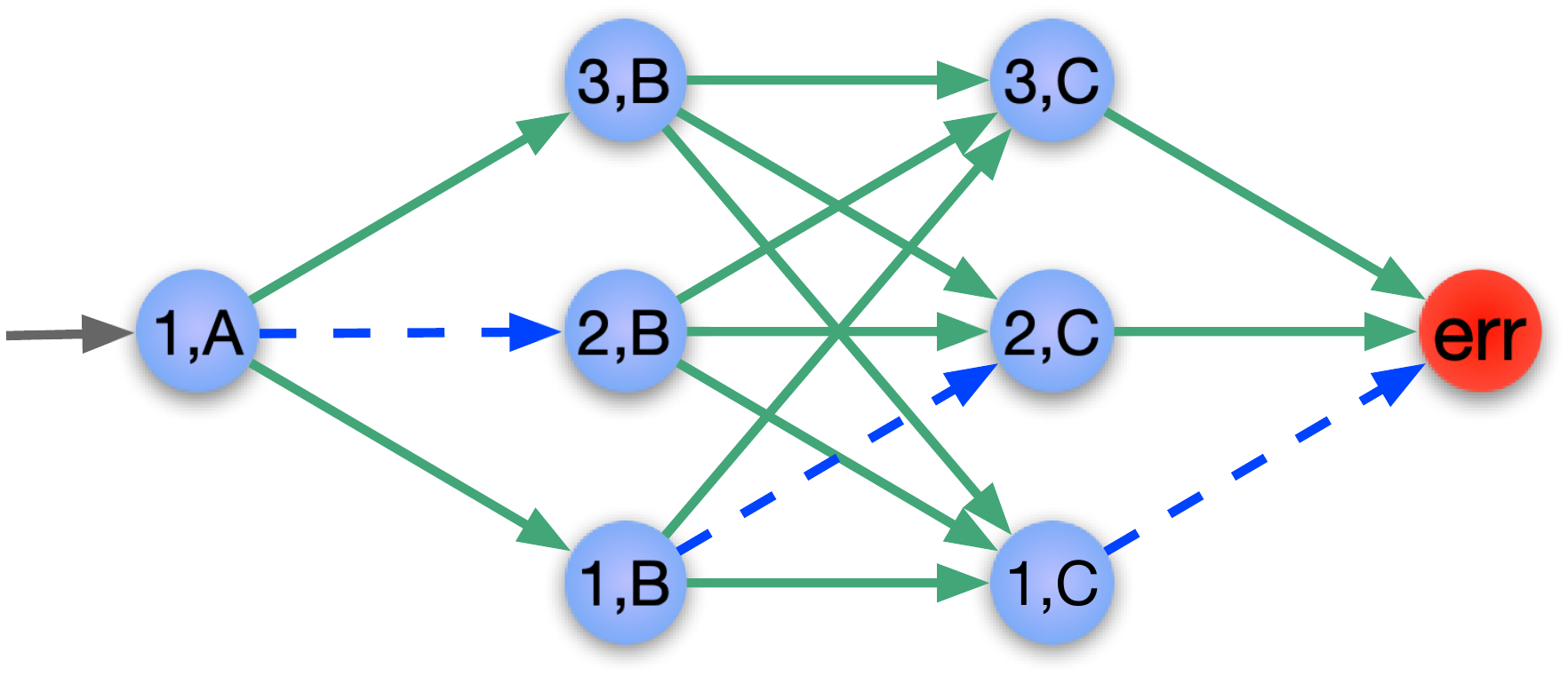}
        \caption{Meta-system $F$}
        \label{fig:meta-system}
\end{subfigure}
\enskip
\begin{subfigure}[b]{0.44\textwidth}
        \centering
        \includegraphics[width=\textwidth]{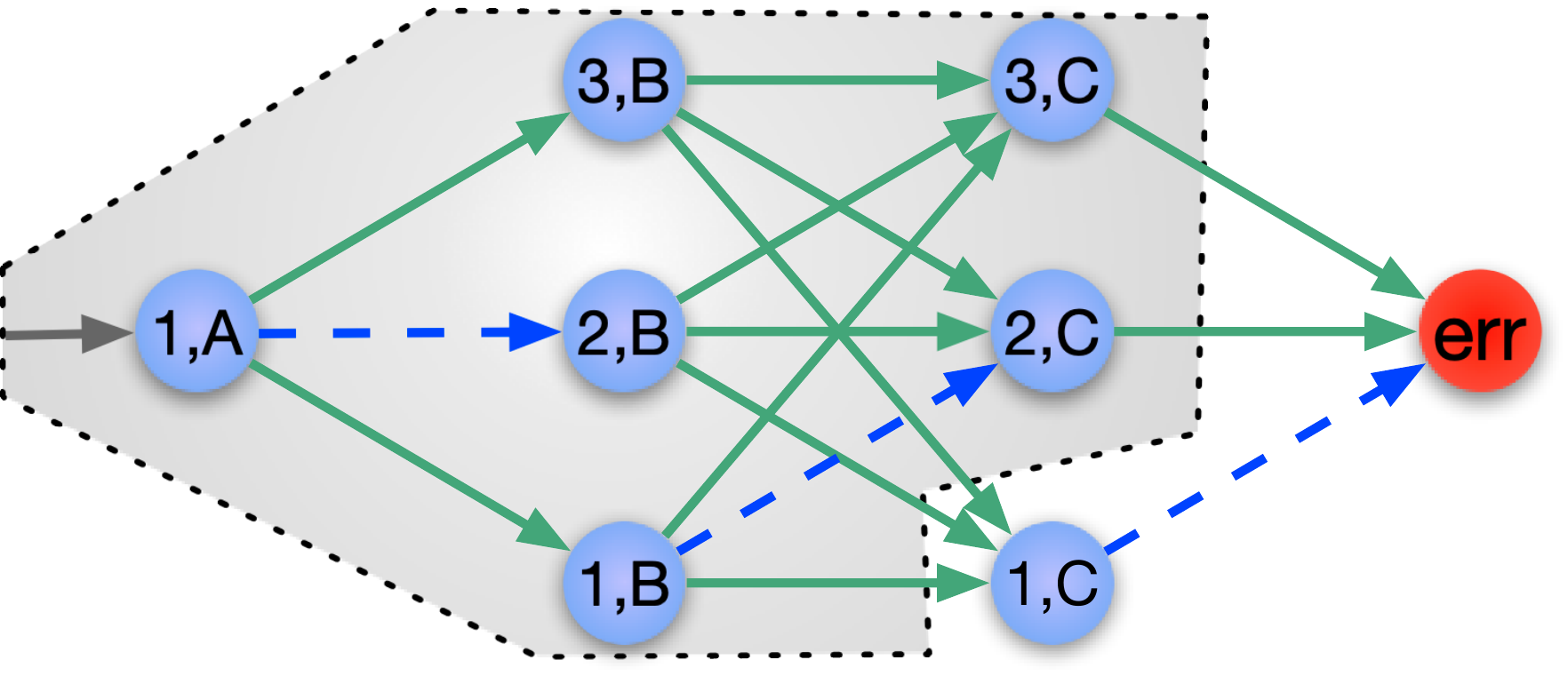}
        \caption{Meta-controller $T_1$}
        \label{fig:meta-system-safe}
\end{subfigure}
\enskip
\begin{subfigure}[b]{0.44\textwidth}
        \centering
        \includegraphics[width=\textwidth]{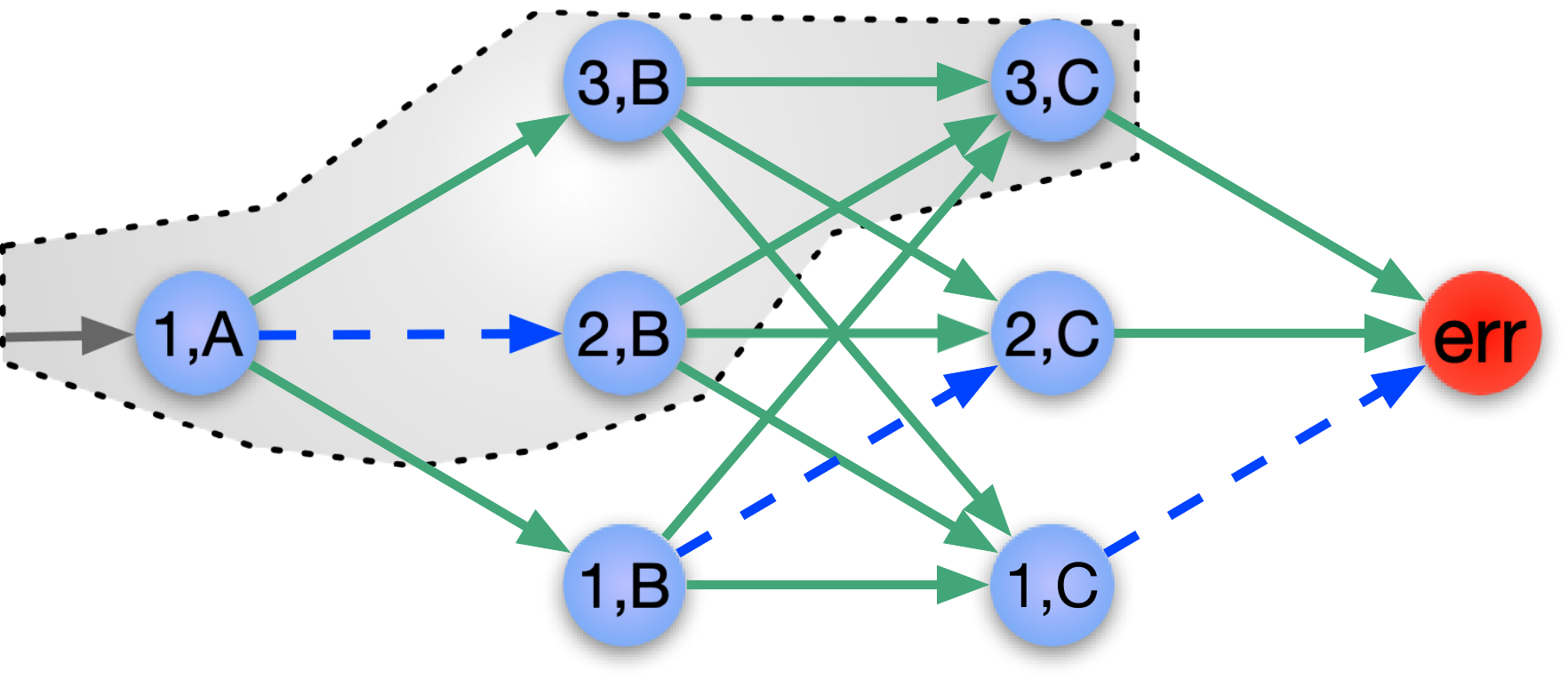}
        \caption{Meta-controller $T_2$}
        \label{fig:meta-system-safe2}
\end{subfigure}

\caption{Meta-systems.  All transitions have action $a$ since $C$ only enables action $a$ (see Fig.~\ref{fig:ctr-alg-lts}). Dashed blue transitions represent transitions that are feasible in $R_E$ while solid green transitions represent the deviated transitions in $(\Alltrans)\setminus R_E$. The shaded area in Fig.~\ref{fig:meta-system-safe} contains all safe states in the meta-system.}
\label{fig:meta-systems}
\end{figure}

\subsubsection{Controlling the meta-system:}
Once the meta-system is constructed, we pose a meta-control problem over $F$ to ensure that the meta-system avoids the error states, i.e., states $(q_e,q_c,err)\in Q_E\times Q_C\times Q_{P_{saf}}$.
These error states represent safety violations in the closed-loop system.
For instance, in Fig.~\ref{fig:meta-system}, if transition $(2,C)\rightarrow err$ occurs, then the closed-loop system violates $P_{saf}$ since more than two actions $a$ were executed.
In this meta-control problem, a meta-controller can disable transitions in $F$ that originated from deviations in $E$, i.e., transitions in $(\Alltrans)\setminus R_E$.
\begin{problem}\label{prob:meta-control}
Given meta-system $F$, synthesize a meta-controller $T\subseteq F$ such that (1) for any $(q_e,q_c,q_p)\in Q_T$ then state $q_p\neq err$; and (2) for any $\big((q_e,q_c,q_p),a,(q'_e,q'_c,q'_p)\big)\in R_F\setminus R_T$ such that $(q_e,q_c,q_p)\in Q_T$, it follows that $(q_e,a,q'_e)\notin R_E$.
\end{problem}

Problem~\ref{prob:meta-control} states that the meta-controller is a subset of the meta-system $F$.
We want to maintain the same structure as in $F$ since we need to enforce that the meta-controller does not disable any transition associated with $R_E$.
Condition (1) in Problem~\ref{prob:meta-control} ensures that property $P_{saf}$ is not violated.
On the other hand, condition~(2) guarantees that only transitions assigned to deviations are disabled.

Back to our example, the LTS $T$ described by the shaded area in Fig.~\ref{fig:meta-system-safe} demonstrates a possible meta-controller that satisfies Problem~\ref{prob:meta-control}.
Condition (1) is satisfied since the error state is not included in the shaded area.
With respect to condition~(2), only solid green transitions are disabled.
Figure~\ref{fig:meta-system-safe2} shows another meta-controller.

To solve Problem~\ref{prob:meta-control}, one can solve a safety game over $F$ using fixed-point computation \cite{Gradel:2002,McNaughton:1993}.
Due to space limitations, we demonstrate the solution to this safety game in the Appendix.




\subsubsection{Extracting robust deviations}

Each meta-controller that solves Problem~\ref{prob:meta-control} relates to a robust deviation.
Intuitively, a meta-controller disables deviations that would violate $P_{saf}$.
For instance, the meta-controller $T_1$ shown in Fig.~\ref{fig:meta-system-safe} disables transition $(3,B)\rightarrow (1,C)$, which relates to disabling transition $3\xrightarrow[]{a} 1$ in the environment.
Figure~\ref{fig:tol-safe-meta} depicts the deviated environment related to meta-controller $T_1$.
Similarly, Fig.~\ref{fig:tol-safe2-meta} shows the deviated environment associated with meta-controller $T_2$.

To extract a robust deviation from a meta-controller, we have to (1) identify the transitions that the meta-controller has disabled; and (2) project the disabled transitions to transitions ${\Alltrans}$.
Since a meta-controller is a subset of the meta-system, the disabled transitions are obtained by comparing $F$ and $T$.
Intuitively, the disabled transitions are those that escape the shaded area in Fig.~\ref{fig:meta-systems}.
\begin{equation}
Disabled := \{(q,\ a,\ q')\in R_F \mid q\in Q_T \ \land \ (q,\ a,\ q')\notin R_T\}
\end{equation}

For instance, in the case of meta-controller $T_1$, the transition $((1,B),a,(1,C))$ belongs to the $Disabled$ set.
Next, based on the disabled transitions, we project them to transitions in $\Alltrans$, i.e., transitions in the environment.
\begin{equation}\label{eq:del}
del := \{(q_e,\ a,\ q'_e)\in \Alltrans\mid ((q_e,q_c,q_p),a,(q'_e,q'_c,q'_p))\in Disabled\}
\end{equation}
Transitions in $del$ are the transitions to be deleted from $\Alltrans$ such that $(\Alltrans)\setminus del$ is a robust deviation set.
If transitions in $del$ are included in a deviation set, they can cause a violation of property $P_{saf}$.
In the case of $T_1$, the transition $(1,a,1)$ is included in $del$.
If we maintain, for instance, transition $1\xrightarrow[]{a} 1$ as part of a deviation set $d$, then the closed-loop $E_d/C$ violates the property $P_{saf}$ since the path $(1,A)\rightarrow (1,B) \rightarrow (1,C) \rightarrow err$ would be feasible in the meta-controller.

\subsubsection{Computing robustness $\Delta$}

Problem~\ref{prob:meta-control} searches for meta-controllers that guarantee the satisfaction of property $P_{saf}$.
To compute $\Delta$, we need to obtain a finite number of meta-controllers.
Algorithm~\ref{alg:tol-control} formalizes our description in Fig.~\ref{fig:overview-alg}.
It takes as input the environment $E$, the controller $C$, a deviation set $d$, and a safety property $P$.
From the algorithm overview description in Fig.~\ref{fig:ex-algo-constr}, we have that for the unconstrained environment $d = A = Q_E\times Act_E\times Q_E$ and $P = P_{saf}$.
\begin{algorithm}
\caption{COMPUTE-ROBUSTNESS}
\label{alg:tol-control}
\begin{center}
\renewcommand{\algorithmicrequire}{\textbf{Input:}}
\renewcommand{\algorithmicensure}{\textbf{Output:}}
\begin{algorithmic}[1]
\vspace*{-0.2cm}
\Require{LTSs $E$, $C$, $P$ and deviation $d$}
\Ensure{Set of deviations $D$}
\State $D \leftarrow \emptyset$ 
\State $F\leftarrow E_d||C||P$
\State $Err \leftarrow \{(q_e,q_c,q_p)\in Q_F\mid q_p = err\}$
\State $W \leftarrow Inv(Q_F\setminus Err)$
\ForAll{$S\in 2^W\setminus \{\emptyset\}$}
\State $T \leftarrow \textsc{Meta-Controller}(S,F)$
\State $del \leftarrow \{(q_e,a,q'_e)\in d\ \mid \exists ((q_e,q_c,q_p),\ a,\ (q'_e,q'_c,q'_p))\in R_F\setminus R_T \text{ s.t. } (q_e,q_c,q_p)\in Q_T\}$
\State $D \leftarrow D \cup \{d\setminus del\}$
\EndFor
\While{$\exists d_1,d_2\in\Delta$ s.t. $d_1\subseteq d_2$} 
\State $D \leftarrow D\setminus \{d_1\}$
\EndWhile
\Return $D$ 
\Procedure{Meta-Controller}{$S,F$}
\State $S\leftarrow Inv(S)$
\If{$q_{0,F}\notin S$}
\State $T \leftarrow \emptyset$
\Else
\State $Q_T \leftarrow S$, $Act_T \leftarrow Act_F$, $q_{0,T} \leftarrow q_{0,F}$ 
\State $R_T \leftarrow\{(q,a,q')\in S\times Act_T\times S\mid (q,a,q')\in R_F\}$
\EndIf
\Return $T$
\EndProcedure
\end{algorithmic}
\end{center}
\vspace*{-0.2cm}
\end{algorithm}

In Alg.~\ref{alg:tol-control}, line 4 computes the largest possible set of invariant states that avoid the error state, i.e., $Inv(Q_F\setminus Err)$ solves the safety game as shown in the Appendix.
Based on this invariant set, each iteration in the loop (lines 5-8) computes a meta-controller (line 6) and stores its respective robust deviation (line 8).
The meta-controller $T$ is also computed by using the function $Inv$.
The meta-controller solution ensures that $Q_T\subseteq S$.
Line 7 computes environmental transitions that must be deleted in order to obtain a robust deviation.
The computed robust deviations are stored in $\Delta$.
Lastly, the loop in lines 9-10 ensures that only maximal robust deviations are included in $\Delta$.

In more detail, to solve Problem~\ref{prob:meta-control}, we must guarantee that the meta-system $F$ does not reach any states in $Err := \{(q_e,q_c,q_p)\in Q_F\mid q_p = err\}$.
Formally, we compute the set $Inv(Q_F\setminus Err)$, which contains every state in $F$ that does not reach a state in $Err$ via a transition associated with $R_E$.
Based on this invariant set, we can extract any meta-controller that remains within this set.
Informally, the $\textsc{Meta-Controller}(S,F)$ in line 11 of Alg.~\ref{alg:tol-control} computes a meta-controller that remains within states in $S$.
First, this procedure computes the invariant set of $S$, i.e., $Inv(S)$ with respect to meta-system $F$ (line 12).
In this manner, a meta-controller is defined by projecting the meta-system $F$ to states and transitions in the set of state $Inv(S)$ (lines 16-17).

The following theorem shows that $\Delta$ computed via Alg.~\ref{alg:tol-control} is equal to $\Delta$ as in Def.~\ref{def:robustness} when $P_{env} = Act_E^*$, i.e., Alg.~\ref{alg:tol-control} \emph{partially} solves Problem~\ref{prob:comp-tol}.

\begin{restatable}{theorem}{theodelta}\label{theo:robustness}
Given LTS $E$, controller $C$, and property $P_{saf}$, Algorithm~\ref{alg:tol-control} outputs $\Delta$ as in Def.~\ref{def:robustness} when $P_{env} = Act_E^*$.
\end{restatable}
\begin{proof}
Sketch.
In order to show that Theorem~\ref{theo:robustness} holds, we provide two intermediate lemmas in the Appendix.
The first lemma states that every meta-controller $T$ produces a robust deviation.
In this manner, we show that for every $d\in \Delta$, the deviation $d$ is robust.
The second lemma shows that for every maximal robust deviation $d\in \Delta$, there exists a meta-controller $T$ associated with deviation $d$.
Consequently, Alg.~\ref{alg:tol-control} computes every possible maximal robust deviation. 
\end{proof}

Using Alg.~\ref{alg:tol-control} to compute $\Delta$ for our running example, we obtain $\Delta$ that contains the three maximal robust deviations shown in Fig.~\ref{fig:robust-envs}.
Lastly, we provide the computational complexity of Alg.~\ref{alg:tol-control}.

\begin{restatable}{theorem}{theodelta}\label{theo:robustness-complexity}
Algorithm~\ref{alg:tol-control} outputs $\Delta$ in $O(2^{|Q_E||Q_C|(|Q_{P}|-1)})$.
\end{restatable}
\begin{proof}
It follows from the size of $2^W$.
\end{proof}
Although Alg.~\ref{alg:tol-control} has exponential complexity, we empirically show in Section~\ref{sect:experiments} that it scales better than the brute-force algorithm.


\subsubsection{Heuristics to exploit the structure of $F$}
In Alg.~\ref{alg:tol-control}, we compute robust deviations for every possible subset of the largest invariant state set, c.f., line 5.
To improve the efficiency of Alg.~\ref{alg:tol-control}, we provide a sound and complete heuristic that identifies and skips redundant subsets of $2^W \setminus \emptyset$.
The heuristic is based on the observation that sets of states that are not directly connected in $F$ correspond to redundant deletion sets from $Q_E \times Act_E \times Q_E$.
As such, the heuristic exploits the structure of $F$ by performing a depth-first search over its state space, hence skipping disconnected groups of states.
For instance, the heuristic will skip the subset $\{(1,A),(3,C)\}$ because $(1,A)$ and $(3,C)$ are not connected in $F$.
This subset is redundant because its deletion set $del = \{((1,A),(1,B)),((1,A),(2,B)),((1,A),(3,B))\}$ is identical to the deletion set for the subset $\{(1,A)\}$ which is connected.
In the worst-case scenario, our heuristic computes the power set of $W$, i.e., exactly as in line 5.


\subsection{Controlling the deviations with environmental constraints}
When introducing environmental constraints, we must eliminate the robust deviations that violate these constraints as described in Def.~\ref{def:robustness}.
One might think that $P_{env}$ and $P_{saf}$ could be combined as a single safety property for which we then compute $\Delta$.
However, this approach does not work since $P_{env}$ must be enforced only by the environment whereas $P_{saf}$ is a property of the closed-loop system.
Another approach is to verify if $P_{env}$ is satisfied for each deviation obtained in the for-loop (lines 5-8) in Alg.~\ref{alg:tol-control}.
Although this approach is feasible, in practice, we want to reduce the number of deviations, using $P_{env}$, before we compute the robust deviations.
For this reason, we describe a sequential algorithm shown in Fig.~\ref{fig:overview-alg-env}.
In this algorithm, Alg.~\ref{alg:tol-control} is used multiple times in this constrained scenario instead of a single time as in the unconstrained scenario (Sect.~\ref{subsect:unconstraint}).
 
\begin{figure}[!t]
\centering  
\includegraphics[width=0.9\textwidth]{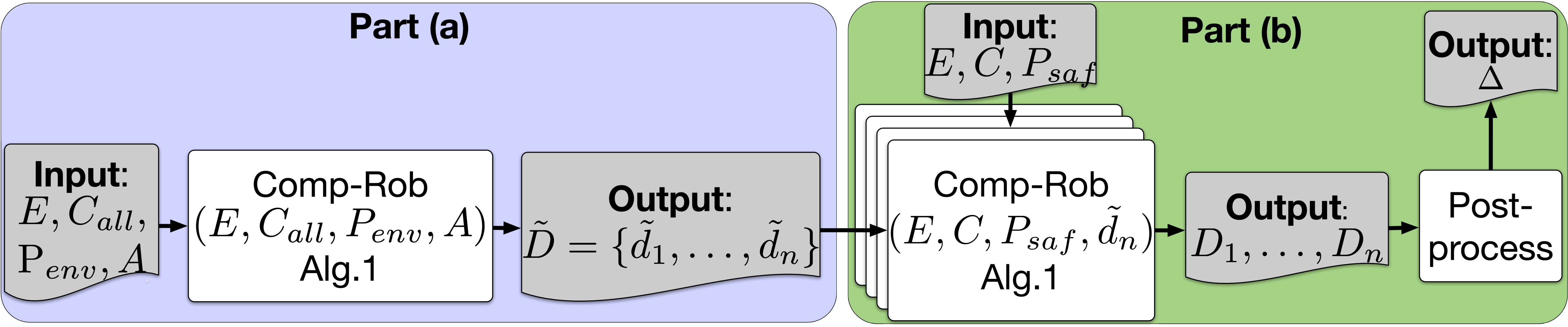}
\caption{Overview of our approach to compute robustness for constrained environments.}
\label{fig:overview-alg-env}
\vspace*{-0.7cm}
\end{figure}

The algorithm to compute robustness for constrained environments can be broken into two parts: (a) computing all maximal environments $\tilde{d}_i$ that satisfy $P_{env}$; and (b) computing robust deviations for each deviated environment $E_{\tilde{d}_i}$ found in part (a).
Computing the maximal environments that satisfy $P_{env}$ reduces to computing maximal deviations of $E$ with respect to a controller that allows every environment action, $C_{all}$.
Formally, the behavior of $C_{all}$ does not restrain $E$, $beh(C_{all}) = Act_E^*$; and it can be described by a one-state LTS.
Therefore, the output of part (a) is the set of maximal deviations $\tilde{d}_i$ with respect to $E$, $C_{all}$, and $P_{env}$, denoted as maximal environment deviations.
Each maximal deviated environment $E_{\tilde{d}_i}$ satisfy the $P_{env}$.

Once we have obtained all maximal environment deviations that satisfy $P_{env}$, we focus on finding the maximal robust deviations with respect to $C$ and $P_{saf}$.
In other words, we run Alg.~\ref{alg:tol-control} for each maximal deviated environment $E_{\tilde{d}_i}$ together with $C$ and $P_{saf}$.
Since $d$ is a subset of $\tilde{d}_i$, we have that the perturbed system $E_d$ satisfies $P_{env}$.

Each maximal deviated environment $E_{\tilde{d}_i}$ generates a set of maximal robust deviations $D_i$ with respect to $C$ and $P_{saf}$.
The final step is combining these maximal robust deviations with respect to each $\tilde{d}_i$.
Since they are maximal with respect to $\tilde{d}_i$, there could be deviations that are not maximal as defined by Def.~\ref{def:robustness}.
The post-processing step combines the deviations and eliminates any non-maximal deviations; and it outputs $\Delta$ as in Def.~\ref{def:robustness}.
The correctness of this algorithm follows from Theorem~\ref{theo:robustness}.

\section{Case studies} \label{sect:experiments}

\subsection{Implementation}\label{sect:impl}

We have implemented a prototype tool for computing robustness \cite{tr-rob-github}.
The tool accepts a model of an environment, a controller, and a safety property--as well as an optional list of environmental constraints--and outputs $\Delta$.
The tool has support for comparing the robustness of two controllers as well as the robustness of a controller with respect to two separate safety properties.
Currently, the environment, controller, safety property, and environmental constraints must be encoded in Finite State Process (FSP) notation \cite{Magee2000} but this is not a fundamental limitation.

We wrote the tool in the Kotlin programming language. 
Our tool includes an implementation of the brute-force algorithm from Sect.~\ref{sect:brute-force}, as well as an implementation of Alg.~\ref{alg:tol-control} and Alg.~\ref{alg:tol-control} with heuristics.
In the following case studies, we leverage the tool to calculate and compare the robustness of several systems.
We summarize our performance results for each case study in Sect.~\ref{sect:perf-results}.

\subsection{Therac-25} \label{sect:therac}

\subsubsection{Background}

In Sect.~\ref{sect:motivating}, we introduced the Therac-25 radiation therapy machine.
In this section, we present a case study in which we compare the robustness of the Therac-25 to that of its predecessor, the Therac-20.
We  begin by showing that the Therac-20 is strictly more robust than the Therac-25.
We then use this information to identify and fix a critical safety bug in the Therac-25 model.

\subsubsection{Therac-20}

The Therac-20 is a radiation therapy machine that was designed before the Therac-25.
Unlike the Therac-25, the Therac-20 was not known for causing accidents that led to injuries and death.
A key difference between the two machines is that the Therac-20 includes hardware \textit{interlocks} in its beam component (Fig.~\ref{fig:therac-beam-interlock}), while the Therac-25 does not (Fig.~\ref{fig:therac-beam}).
The purpose of the hardware interlocks is to provide a layer of security at the hardware level for upholding $P_{xflat}$.
In our model, the interlocks work by ensuring that the flattener is completely rotated into place before allowing an operator to fire an X-ray beam.
Unfortunately, hardware interlocks were considered expensive so they were omitted from the design of the later Therac-25 model.
In the following section, we compare the robustness between the two Therac machines with respect to the normative environment $E$ and the key safety property $P_{xflat}$.

\begin{figure}[!t]
\begin{subfigure}[b]{0.48\textwidth}
    \centering
    \includegraphics[height=100px]{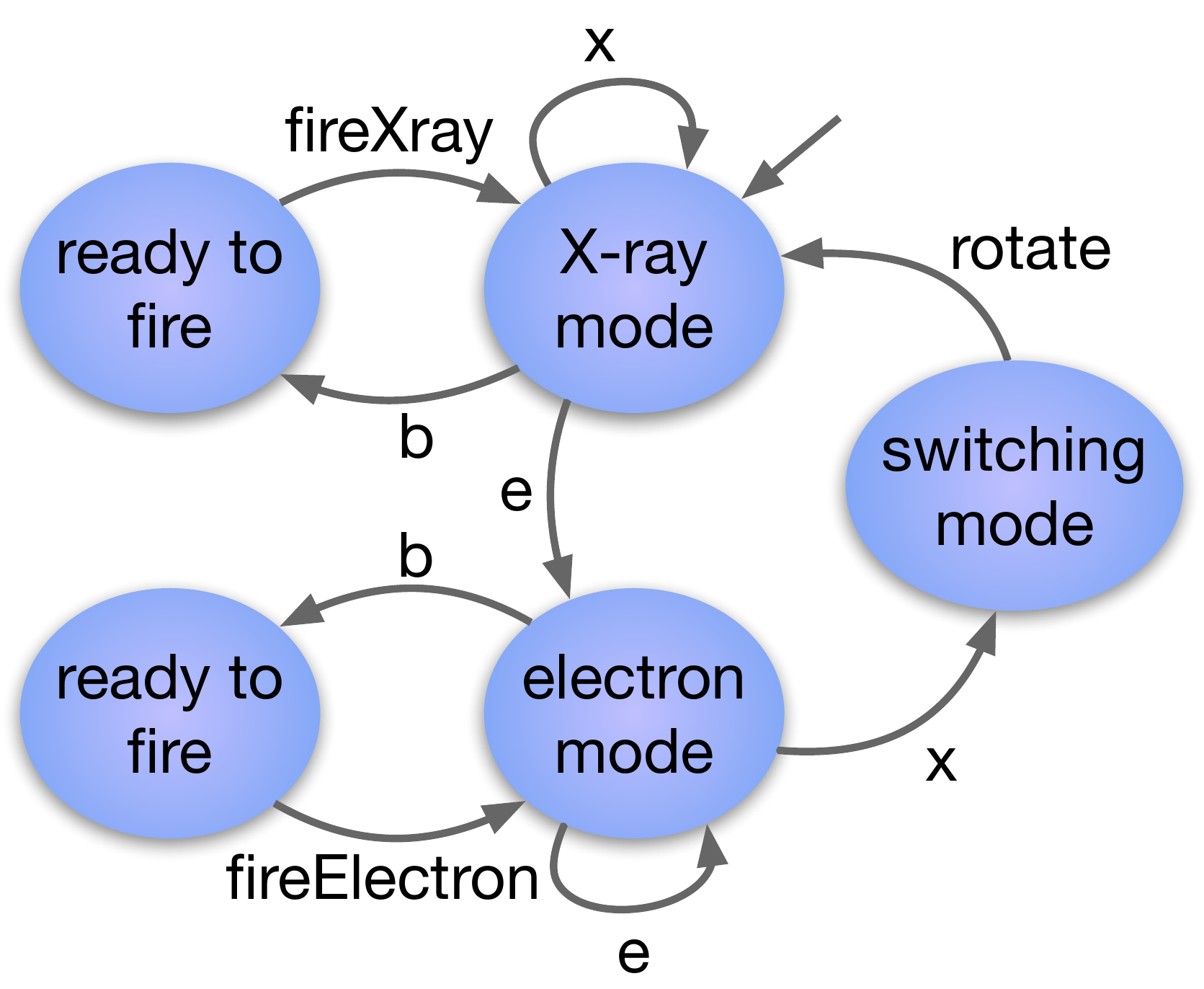}
    \caption{The beam $C'_{beam}$ \textit{with} hardware interlocks used in the Therac-20.}
    \label{fig:therac-beam-interlock}
\end{subfigure}
\enskip
\begin{subfigure}[b]{0.48\textwidth}
    \centering
    \includegraphics[height=100px]{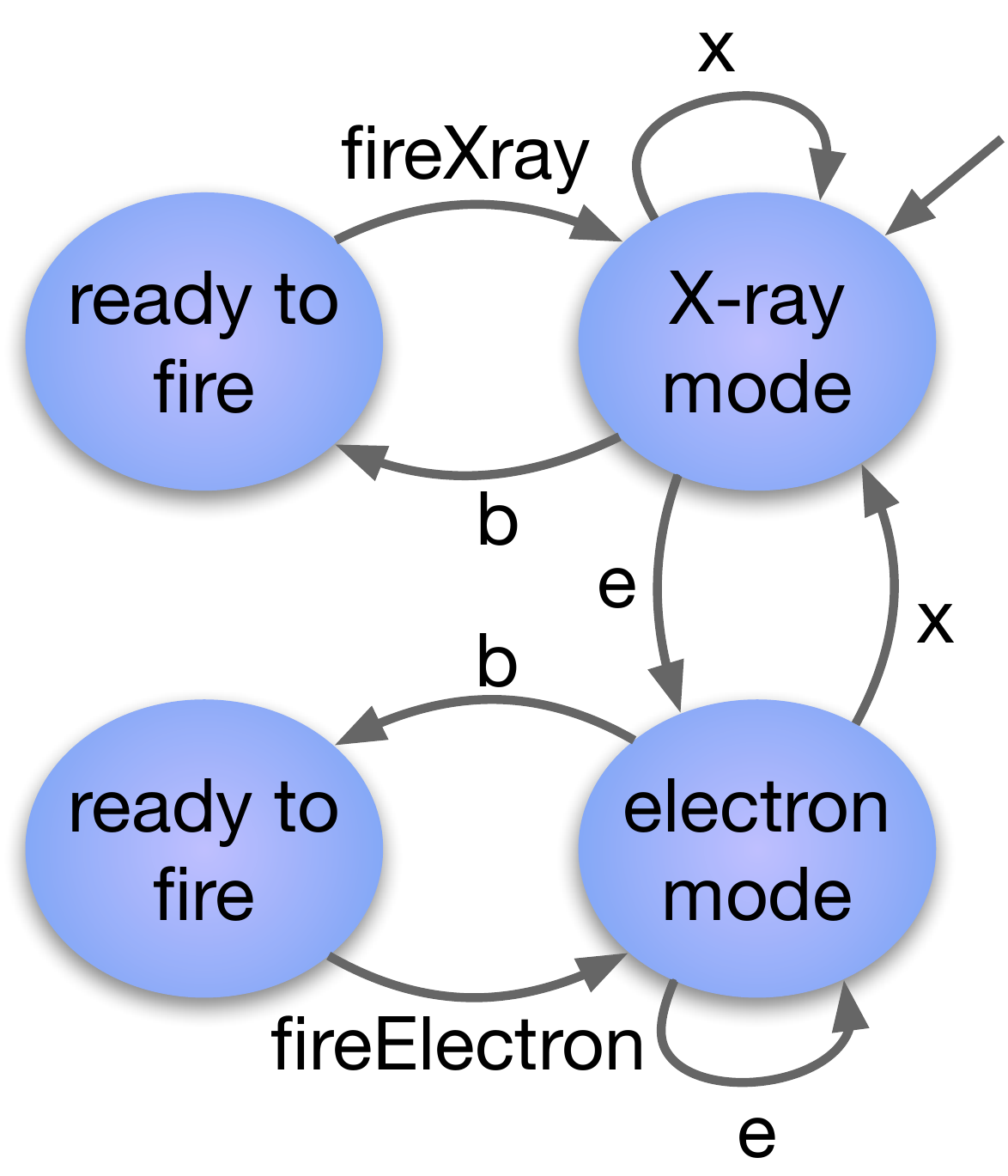}
    \caption{The beam $C_{beam}$ \textit{without} hardware interlocks used in the Therac-25.}
    \label{fig:therac-beam}
\end{subfigure}

\caption{The beam components of the two Therac machines.
The hardware interlocks cause $C'_{beam}$ to have a fifth state ``switching mode'' that will only switch to X-ray mode after the flattener rotates into place.}
\vspace{-20px}
\end{figure}

\subsubsection{Comparing controllers}

Using standard model checking techniques \cite{Baier:2008}, we can confirm that both the Therac-20 and the Therac-25 are safe with respect to $E$ and $P_{xflat}$.
Historically, however, the Therac-20 is known to be safer than the Therac-25.
Therefore, we improve our safety analysis by also comparing the robustness between the two machines with respect to $E$, $P_{xlfat}$, and an environmental constraint $P_{env}$.
$P_{env}$, shown in Fig.~\ref{fig:p-env-therac} in the Appendix, restricts the environment to firing the beam at most once.

Our tool reports that the Therac-20 is strictly more robust than the Thearc-25.
To understand this result, we can examine the difference between the robustness for each machine.
We show this difference visually by presenting one maximal robust deviation from each machine in Fig.~\ref{fig:therac-tol}.
This figure shows that the Therac-20 is robust against the scenario in which the operator
1) types ``e'' to select electron beam mode,
2) optionally types ``enter'',
3) presses the ``up'' arrow key, and finally
4) types ``x'' to switch the beam into X-ray mode.
The Therac-25, however, is not robust against this scenario.
We see this in Fig~\ref{fig:therac-tol} because the series of actions must pass through at least one green arrow, where a green arrow indicates a transition that the Therac-25 is not robust against.
In fact, the Therac-25 does not have \textit{any} maximal robust deviations that allow this scenario.

The Therac-25's lack of robustness to the scenario above represents a race condition that occurs after the operator switches into X-ray mode from electron mode.
In this scenario, if the operator types ``enter'' and fires the X-ray beam before the flattener rotates into place, the beam will fire an unflattened X-ray at the patient.
This critical bug was responsible for real-world radiation overdoses, several of which resulted in death \cite{Leveson:1993}.

\subsubsection{Fixing the software bug}

In the previous section, we identified a critical software bug in the Therac-25.
Our goal in the current section is to fix this bug entirely in the terminal software, thus avoiding an expensive hardware solution.

In Fig.~\ref{fig:therac-beam-interlock}, we see that the hardware interlocks prevent a race condition by blocking the operator from typing a ``b'' until the flattener is rotated into place.
Thus we can fix the race condition in software by altering the terminal to block the operator from typing a ``b'' until the flattener is rotated into place. 
We implement this fix by redesigning the terminal to block all key strokes from the instant it issues a ``beam ready'' message until the turntable rotates into place, as shown in Fig.~\ref{fig:therac-fix}.
Finally, we use our tool to evaluate the robustness of the fix.
The tool reports that the fixed Therac-25 design is strictly more robust than the original, and equally robust to the Therac-20.

\begin{figure}[!t]
\begin{minipage}[c]{0.48\linewidth}
    \centering
    \includegraphics[height=75px]{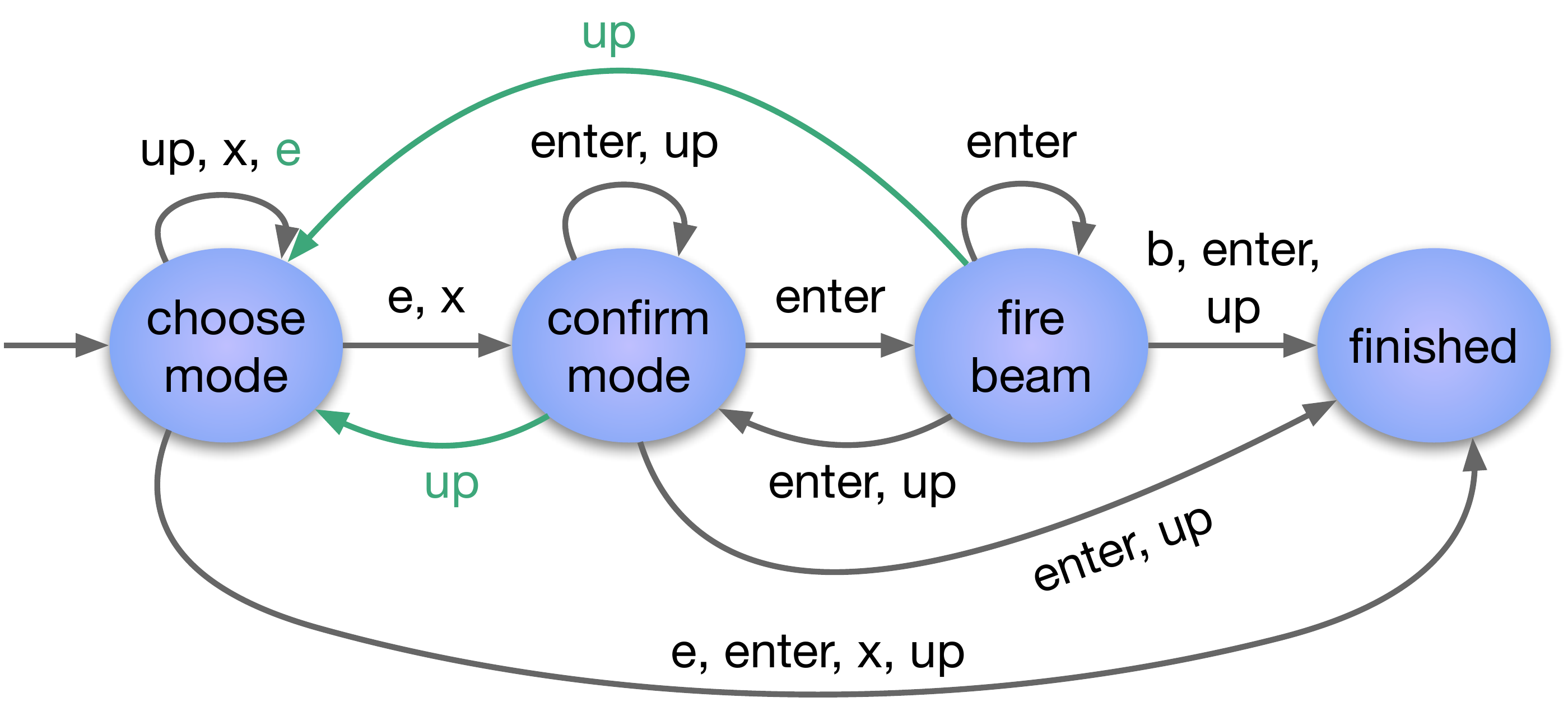}
    \caption{
        Visual robustness comparison between the two Therac machines.
        Both machines are robust against gray transitions, but only the Therac-20 is robust against green transitions.
    }
    \label{fig:therac-tol}
\end{minipage}\hfill
\begin{minipage}[c]{0.48\linewidth}
    \includegraphics[height=80px]{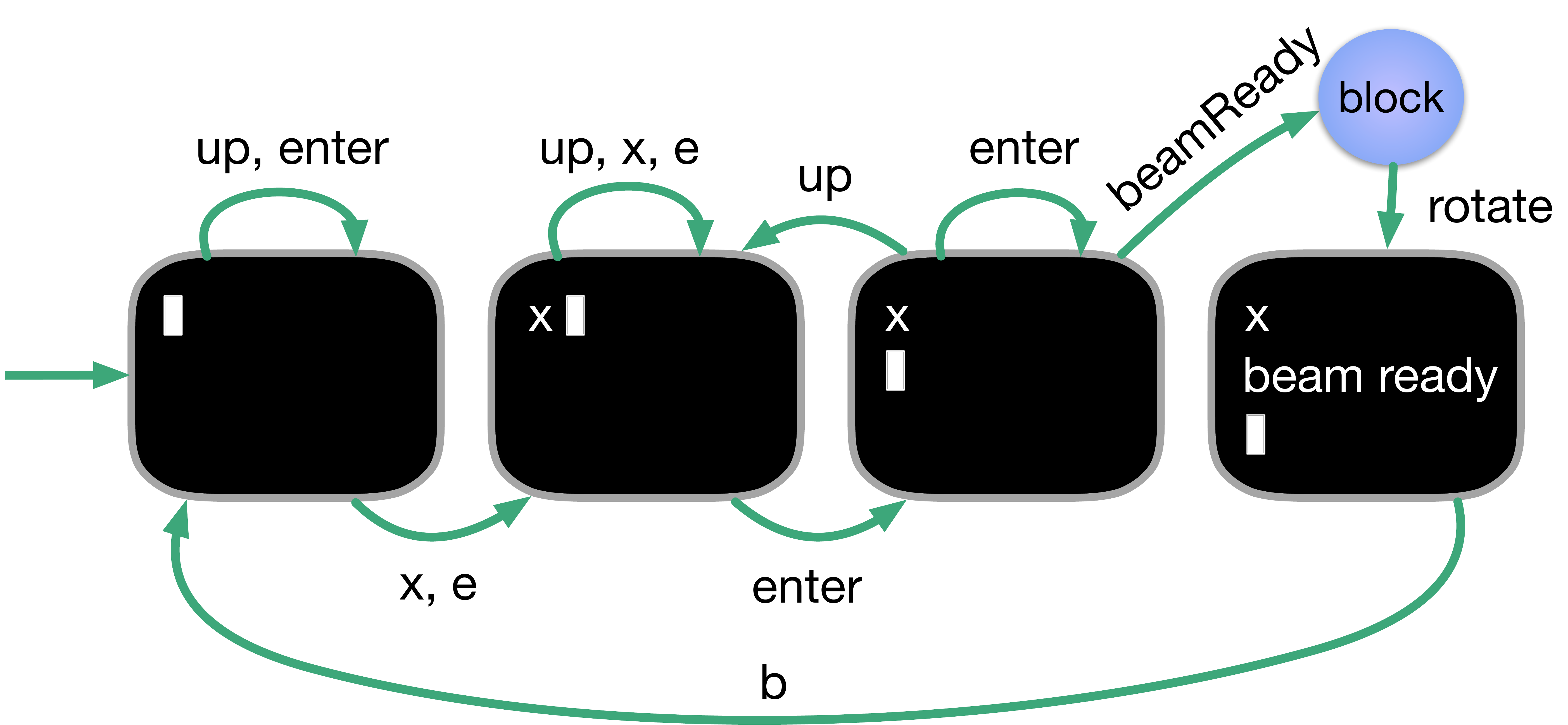}
    \caption{Software fix that eliminates the race condition in the Therac-25.
    }
    \label{fig:therac-fix}
\end{minipage}
\vspace{-20px}
\end{figure}

\subsection{Voting} \label{sect:voting}

\subsubsection{Background}

In this section, we consider a case study of an electronic voting machine, introduced in \cite{Tun2020}.
In this case study, we model the voting machine, a voter, and a corrupt election official who attempts to ``flip'' the voter's choice.
We define the voting machine as the composition of a voting booth and a user interface, shown in the Appendix in Figs. \ref{fig:voting-booth} and \ref{fig:voting-machine-screen} respectively.

In the normative environment--shown in Fig.~\ref{fig:voting-env}--the voter enters the booth, enters their password, selects a candidate, clicks the vote button, and finally confirms the choice.
Unfortunately, some voters may inadvertently skip the confirmation step and leave the booth early.
This deviation from the normative behavior presents an opportunity for the election official to ``flip'' the intended vote:
after the voter leaves the booth, the corrupt official can enter the booth, press ``back'' and change the vote to their liking. This scenario represents an actual election fraud that took place in the US~\cite{fbi_2010}.

\begin{figure}[!t]
\begin{subfigure}[b]{0.45\textwidth}
    \centering
    \includegraphics[height=60px]{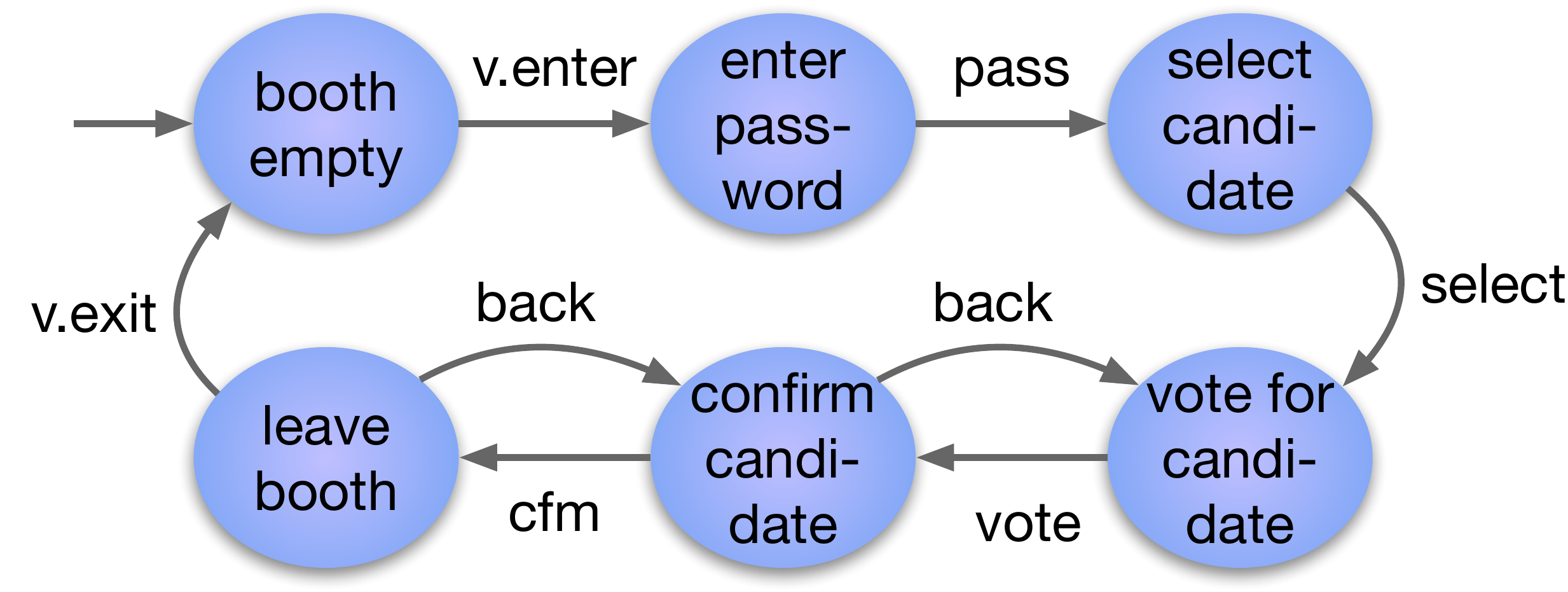}
    \caption{Normative environment for the voting machine.}
    \label{fig:voting-env}
\end{subfigure}
\enskip
\begin{subfigure}[b]{0.51\textwidth}
    \centering
    \includegraphics[height=90px]{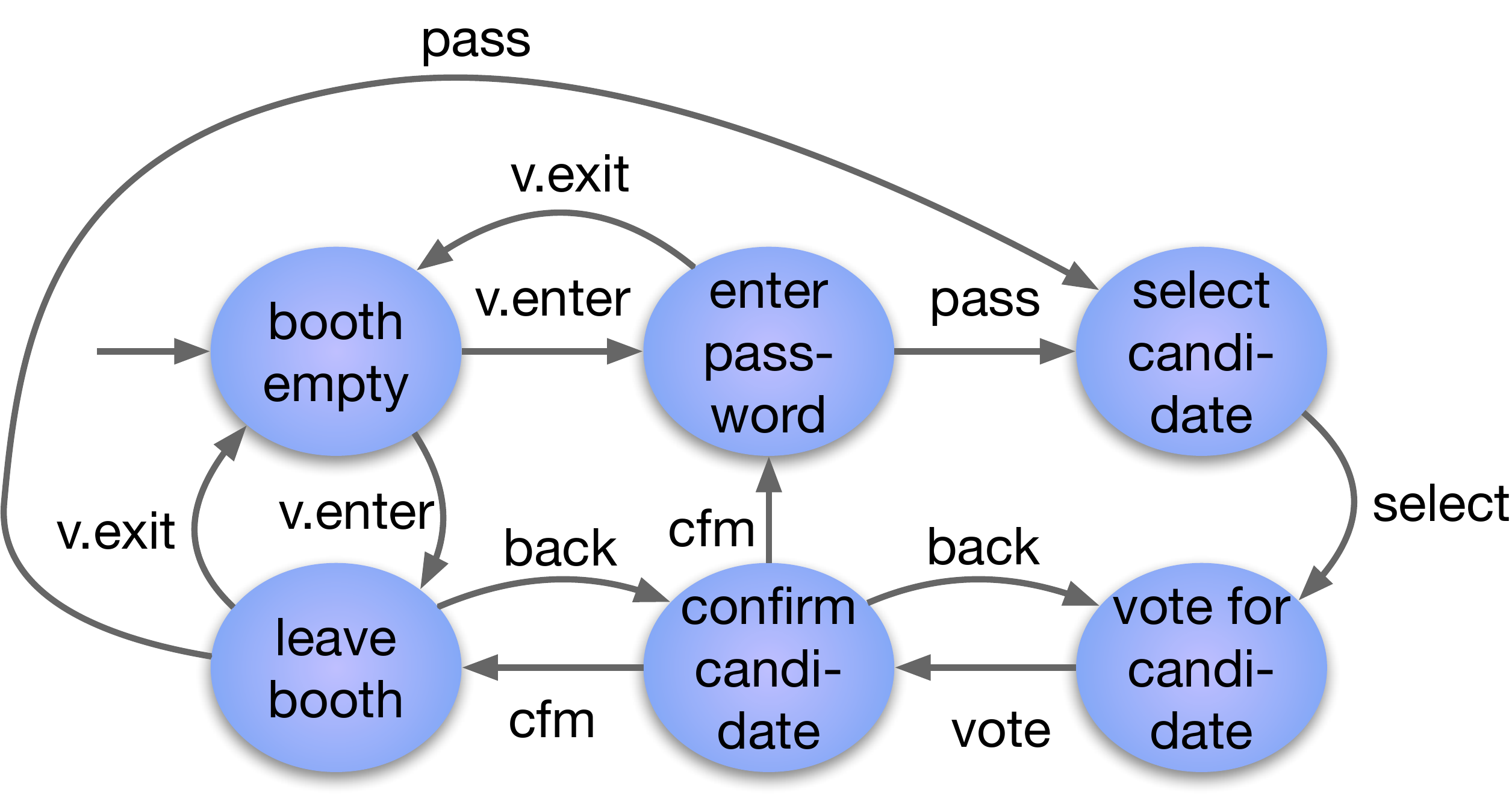}
    \caption{The voting machine's robustness is identical with respect to $P_{all}$ and $P_{cfm}$.}
    \label{fig:voting-tolerance}
\end{subfigure}

\caption{
Models for the voting machine example.
In the figures above, the prefix ``v'' represents actions by the voter.}
\vspace{-20px}
\end{figure}

\subsubsection{Comparing properties}

In this case study, we will consider two safety properties, $P_{all}$ and $P_{cfm}$, both of which imply the absence of vote flipping.
$P_{all}$ requires that the election official cannot at any point select, vote, or confirm a candidate.
$P_{cfm}$ is weaker, only requiring that the election official cannot at any point confirm a candidate selection.

Using our tool for comparison, we see that the voting machine is equally robust with respect to each property.
However, this result is surprising because $P_{cfm}$ is weaker than $P_{all}$.
To understand this result, we examine Fig.~\ref{fig:voting-tolerance} where we present the sole maximal robust deviation for each property.
In this figure, it is clear that the voting machine is not robust against any deviation in which the voter enters their password and then exits the booth without confirming their vote.
The key insight is that, when an election official has the ability to confirm, it \textit{implies} that the official can also select and vote.
Therefore, we desire a voting machine without this implication because it will reduce the number of points of failure.
For example, we could redesign the voting machine to require a password as part of the confirmation step.
In lieu of this insight, a designer could choose to specify a margin of safety into the machine's specification by requiring that it is strictly more robust against $P_{cfm}$ than $P_{all}$.

\subsection{Oyster}

\subsubsection{Background}

The Oyster example was introduced in \cite{Sempreboni:2020}, in which the authors modeled the Oyster card that is used the public transportation system in the United Kingdom.
In our model, the controller consists of an \textit{entry gate} and an \textit{exit gate}, where the card holder taps the Oyster card at the start and end of their journey respectively.
The environment models the actions of a card holder; in the normative environment, a card holder chooses to tap with either their Oyster card or a credit card, and taps in and out with the chosen card.
The key safety property is avoiding an \textit{incomplete journey}, in which a card holder taps in with one card and taps out with a different card.

\subsubsection{Calculating Robustness}

An incomplete journey is avoided under the normative environment.
We calculate the robustness of the system under the two environmental constraints
1) Oyster cards and credit cards give the correct information to the gates and
2) the gates operate correctly and calculate the correct fare when a card is tapped in and out.
Unfortunately, the system is not robust to \textit{any} deviations.

\subsection{PCA Pump}

\subsubsection{Background}

In this section, we model a patient-controlled analgesia (PCA) pump, originally introduced in \cite{Bolton:2011}.
A PCA pump is a medical device that dispenses pain medicine to a patient, offering them partial control over the dose rate.
A nurse uses the device interface to program the volume per dosage, as well as a minimum and maximum dose rate to protect the patient from an overdose.
The pump includes batteries to power the device in case it is unplugged (e.g., by mistake by the nurse or patient), yet the power may fail if the device runs out of battery.
In this case, the device cannot monitor the dosage amount or frequency, which may cause an overdose.
Therefore, we define the key safety property $P_{pfail}$ which requires the PCA pump to abstain from administering medicine after a power failure.

In the normative environment, the nurse operates the pump using the following three step workflow:
1) plug in the pump and turn it on,
2) program the desired dosage parameters into the pump and administer the treatment, and
3) turn off the device and unplug it.
The nurse begins with step (1) and ends with step (3), but may omit or repeat step (2) as many times as needed.
We show the normative environment in Fig.~\ref{fig:pump-env} in the Appendix.
Crucially, the pump is safe with respect to this environment and $P_{pfail}$ because the workflow assumes that the pump is never unplugged in step (2).

\subsubsection{Calculating Robustness}

We use our tool to calculate the robustness of the pump with respect to the normative environment, $P_{pfail}$, and an environmental constraint $P_{env}$.
In this case study, $P_{env}$ restricts the environment to actions that are allowed by the pump's interface.
The sole maximal robust deviation is shown in Fig.~\ref{fig:pump-tol} in the Appendix.
The tool reports that the pump is robust against four actions, three of which allow the operator to change settings before administering the treatment, and the fourth allows the operator to turn off the device prematurely after programming the dosage parameters.
Unfortunately, the pump is not robust against any deviations in which it is unexpectedly unplugged.
This poses a key weakness in the pump that the designers may wish to improve upon.


\subsection{Results and Discussion} \label{sect:perf-results}
We have run our tool on the examples and case studies above, and we present our results in Table \ref{table:results}.
All tests were run on a Mac Book Pro with an M1 Pro chip and 32GB of RAM.
In the table, $|Act|$ is the union of $Act_E$, $Act_C$, $Act_{P_{saf}}$ and $Act_{P_{env}}$,
$|d_{max}|$ is the size of the largest deviation in $\Delta$,
and $|W_{P_{env}}|$ is the size of the winning set for each maximal deviation $\tilde{d}_i$ (separated by a comma); NA indicates the absence of an environmental constraint.
Furthermore, ``Wall Heur'' denotes the wall time for running Alg.~\ref{alg:tol-control} with the heuristic, while ``Wall Plain'' denotes the wall time for running Alg.~\ref{alg:tol-control}, and ``TO'' indicates a time-out after five minutes.

Our results demonstrate that calculating robustness is tractable across several different case studies.
In particular, our tool's performance on the larger PCA pump case study shows promising results in terms of scalability.
Furthermore, we have shown that $\Delta$ is useful as a means for both analysis and comparison of controllers.
For example, in the Therac-25 case study, robustness provided a richer analysis than classic verification that helped us discover--and ultimately fix--a critical race condition.
Finally, we have also demonstrated in the voting machine case study that robustness provides a means for comparing two properties with respect to a controller and an environment.



\begin{table}
    \centering
    \begin{tabular}{|l|c|c|c|c|c|c|c|c|c|c|}
        \hline
        \textbf{Example} & $|Act|$ & $|Q_E|$ & $|Q_C|$ & $|Q_P|$ & $|W|$ & $|W_{P_{env}}|$ & $|\Delta|$ & $|d_{max}|$ & Wall Heur & Wall Plain\\
        \hline
        Running Example & 2 & 4 & 2 & 4 & 6 & NA & 3 & 13 & 0.433 sec & 0.431 sec\\
        \hline
        Therac-25 w/bug & 9 & 5 & 21 & 5 & 62 & 28,30,31,37 & 4 & 21 & 4.921 sec & TO\\
        \hline
        Therac-25 w/fix & 9 & 5 & 19 & 5 & 72 & 18,20,23,25 & 4 & 26 & 0.852 sec & TO\\
        \hline
        Therac-20 & 9 & 5 & 11 & 5 & 40 & 17,19,21,23 & 4 & 26 & 0.626 sec & TO\\
        \hline
        Voting wrt. $P_{cfm}$ & 9 & 7 & 13 & 3 & 66 & 7 & 1 & 12 & 0.469 sec & TO\\
        \hline
        Voting wrt. $P_{all}$ & 9 & 7 & 13 & 3 & 66 & 7 & 1 & 12 & 0.426 sec & TO\\
        \hline
        Oyster & 8 & 4 & 17 & 2 & 15 & 8 & 1 & 4 & 0.472 sec & TO\\
        \hline
        PCA Pump & 21 & 11 & 105 & 4 & 1396 & 34 & 1 & 15 & 1.922 sec & TO\\
        \hline
    \end{tabular}
    \vspace{5px}
    \caption{Summary of results from running our tool.}
    \label{table:results}
    \vspace*{-1cm}
\end{table}

\section{Related work} \label{sect:related}
Quantitative robustness notions for discrete transition systems have been investigated in several works \cite{Bloem:2014,Bloem:2009,Chaudhuri:2011,Henzinger:2014,Majumdar:2011,Neider:2020,Roopsha:2013,Tabuada:2012}.
We capture robustness qualitatively, which avoids the need for external cost functions over the discrete transition systems.
The problem of synthesizing robust controllers against deviated environments given by a designer is investigated in \cite{Topcu:2012}.
Since \cite{Topcu:2012} focuses on synthesizing robust controllers, their framework does not address the analysis of robustness.
Moreover, robust controllers are measured via a rank function (quantitatively).
Robust linear temporal logic (rLTL) extends the binary view of LTL to a 5-valued semantics to capture different levels of property satisfaction \cite{Tabuada:2016}.
This work is tangent to ours as it focuses on specifying robustness.

In \cite{Kang:2020,Zhang:2020}, the authors define robustness as a set of environmental behaviors for which a software system can guarantee safety.
Defining robustness in the semantic domain--i.e. in terms of behaviors--implicitly describes safe environmental deviations.
Our notion of robustness captures safe environmental deviations explicitly in terms of transitions, which offer both syntactic (transitions) and semantic (implied behaviors) information.
Transition-based robustness also allows us to capture the safe environmental envelopes of a system; it is not clear how one might efficiently capture this information with only behaviors.

In \cite{Meira-Goes:2021arxiv}, the authors define robustness also based on additional transitions to the environment. 
Their definition of robustness compares the perturbed controlled behavior, i.e., $beh(E_d|f)$, instead of directly comparing the additional transitions.
In this manner, the partial order used to define robustness in \cite{Meira-Goes:2021arxiv} is different from our notion of robustness.
Moreover, only an efficient algorithm for invariance properties is presented.
Extending the work in \cite{Meira-Goes:2021arxiv}, the authors explore the relationship between controller robustness and permissiveness for invariance properties \cite{Meira-Goes:2022wodes}.

Robust control in discrete event systems is also an active area of research \cite{Alves:2019,Cury:1999,Lin:1993,Lin:2014,Lin:2019cdc,Meira-Goes:2019cdc,Meira-Goes:2021tac-robust,Paoli:2005,Rohloff:2012,Takai:2004,Wang:2016,Young:1995}.
However, they usually deal with specific types of faults such as  communication delays, loss of information, or deception attacks \cite{Alves:2019,Lin:2014,Lin:2019cdc,Meira-Goes:2019cdc,Meira-Goes:2021tac-robust,Rohloff:2012,Wang:2016}.
We capture model uncertainty with our robustness definition, which can be attributed to these faults.
Robustness against model uncertainty is tackled in the works of \cite{Cury:1999,Lin:1993,Takai:2004,Young:1995}.
In these works, deviations are modeled by the behavior generated by the environment.
On the other hand, we modeled deviations by the inclusion of extra transitions.
In \cite{DIppolito:2012}, a controller realizability problem is studied for environments modeled as modal transition systems, where a controller satisfies a property in all, some, or none of the LTS family.
Our notion of robustness explicitly computes which systems in the LTS family satisfy the property.

Lastly, robustness also relates to fault-tolerance.
Fault-tolerance has been studied in the context of distributed systems \cite{Gartner:1999,Lynch:1996,Pease:1980}.
In \cite{Bonakdarpour:2008,Cheng:2011,Ebnenasir:2005,Girault:2009}, synthesis of fault-tolerant programs by retrofitting initial fault-intolerant programs.
These works focus on specific types of fault models, whereas our robustness model computes the safety envelope the controller is robust against. 

\section{Conclusion} \label{sect:conclusion}

In this paper, we introduced a new notion of robustness against environmental deviations for discrete-state transition systems.
Our notion of robustness is syntactically defined by additional transitions and semantically defined by the controlled behavior generated by these additional transitions.
We provided two methods to compute robustness: a brute-force algorithm, and an algorithm based on a controller synthesis problem.
We implemented these methods in a prototype tool which we used to analyze several case studies. 
In these case studies, we demonstrated that our robustness analysis provides crucial information by identifying the environmental envelopes in which the system can guarantee its safety properties.

As part of future work, we plan to extend our work to investigate robustness in the context of partially observable systems as well as in stochastic systems such as Markov decision processes (MDPs).
We also plan to investigate the benefit of considering additional environmental states--as well as additional transitions--in our robustness analysis.
Finally, we plan to extend our work beyond safety properties, e.g. including liveness.

\section*{Acknowledgements} This project was supported by the US NSF Awards CCF-2144860, CNS-1801342, CNS-1801546, CCF-1918140, and ECCS-2144416.
\bibliographystyle{splncs04}
\bibliography{bib_romulo.bib}

\longfalse
\longtrue
\iflong 
\appendix
\section*{Appendix}
\subsection*{Uniqueness of $\Delta$}
\lemmarobustnessuniqueness*
\begin{proof}
By contradiction.
Assume that there exist $\Delta_1,\Delta_2 \subseteq 2^{Q_E\times Act_E\times Q_E}$ such that they satisfy conditions 1, 2, 3, and 4 in Def.~\ref{def:robustness} and $\Delta_1 \neq \Delta_2$.
Without loss of generality, we assume that $\exists d_1\in \Delta_1\setminus\Delta_2$.
Since $d_1\in \Delta_1$, we have that $E_{d_1}/C\models P_{saf}$ and $E_{d_1}\models P_{env}$ as $\Delta_1$ satisfies 1 and 4.
As $\Delta_2$ satisfies 2 and $d_1\notin \Delta_2$, we have that $\exists d_2\in \Delta_2$ such that $E_{d_2}/C\models P_{saf}$,  $E_{d_2}\models P_{env}$, and $d_1\subseteq d_2$.
Since $d_1\in \Delta_1\setminus\Delta_2$, it follows that $d_1\subseteq d_2$.
Back to $\Delta_1$, condition 3 implies that $d_2\notin\Delta_1$ since $d_1\in \Delta_1$ and $d_1\subseteq d_2$.
Furthermore, it does not exist $d\in \Delta_1$ such that $d_2\subseteq d$, because $d_1 \in \Delta_1$ and $d_1\subseteq d_2\subseteq d$, and $\Delta_1$ satisfies condition 3.
Consequently, the deviation $d_2$ is a witness of the $\Delta_1$ violating condition 2, which contradicts our assumption that $\Delta_1$ satisfies conditions 1, 2, 3, and 4.
\end{proof}

\subsection*{Solving safety games}
We now briefly describe how to solve the safety game that provides solutions to Problem~\ref{prob:meta-control}.
We define the  \emph{invariant} of the set of states $S\subseteq Q_F$, denoted as $Inv(S)$.
Intuitively, the invariant of the set of states $S\subseteq Q_F$ contains exactly the states from which the meta-controller can force the meta-system to remain within $S$.
Formally, the invariant is defined recursively using sets $Inv^i(S)$ from which the meta-controller can enforce retention in $S$ for at most $i\geq 0$ steps.
\begin{equation}
Inv^0(S) := S
\end{equation}
\begin{equation}
Inv^i(S) := Inv^{i-1}(S)\ \cap\ \{q\in Q_F\mid R_F \downharpoonright  E(q) \subseteq Inv^{i-1}(S)\}
\end{equation}
and
\begin{equation}
Inv(S) := \bigcap_{i\geq 0} Inv^i(S)
\end{equation}
where $R_F\downharpoonright E\big((q_e,q_c,q_p)\big) =\{(q'_e,q'_c,q'_p)\in Q_F \mid \exists ((q_e,q_c,q_p), a, (q'_e,q'_c,q'_p))\in R_F \text{ s.t. }\allowbreak (q_e,a,q'_e)\in R_E\}$ describes the controllable successor states with respect to $R_E$.

\subsection*{Proof of Theorem~\ref{theo:robustness}}
To prove Theorem~\ref{alg:tol-control}, let us assume the definitions given in Alg.~\ref{alg:tol-control} and Sect.~\ref{subsect:unconstraint}.

\begin{lemma}
Let environment $E$, controller $C$, safety property $P_{saf}$, and environmental property $P_{env} = Act_E^*$ be given.
For any set of states $S\subseteq W$, the deviation $d = A\setminus del $ is robust, where $W = Inv(Q_F\setminus Err)$, $A = Q_E\times Act_E\times Q_E$ and $del$ is defined in Eq.~\ref{eq:del}.
\end{lemma}
\begin{proof}
Let $T$ be the meta-controller found by computing $Meta-Controller(S,F)$.
Note that $T\subseteq F$ and that $T$ satisfy Problem~\ref{prob:meta-control}.
By construction of $d$, we have that $E_d||C||P_{saf}\subseteq T$ since $d$ is constructed based on $T$.  
It follows that for any $x_0\dots x_n \in beh(E_d||C||P_{saf})$, we have that $x_n\notin Err$.
Therefore, $d$ is robust.
\end{proof}

\begin{lemma}
Let environment $E$, controller $C$, safety property $P_{saf}$, and environmental property $P_{env} = Act_E^*$ be given.
For any deviation $d\in \Delta$, there exists a set of states $S\subseteq W$ such that $d= A\setminus del$ where $del$ is defined in Eq.~\ref{eq:del} with set $S$.
\end{lemma}
\begin{proof}
Let $M = E_d||C||P_{saf}$, we show that $d = A\setminus del$ where $del$ is computed based on $Q_{M}$.
Note that $Q_M \subseteq W$ since $d$ is robust.
Let $T = Meta-Controller(Q_M,F)$, we show that $Q_M = Q_T$ and $R_M\subseteq R_T$.

First, we show that $Q_M = Q_T$. 
We start by showing that $Inv(Q_M) = Q_M$.
By the definition of $Inv$, we have that $Inv(Q_M)\subseteq Q_M$.
Next, we demonstrate that $Q_M\subseteq Inv(Q_M)$.
Since $E_d$ contains all transitions in $E$, it follows that $R_F\downharpoonright E(Q_M)\subseteq Q_M$.
Therefore, it must be that $Q_M\subseteq Inv(Q_M)$.

In the second part, we show that $R_M\subseteq R_T$.
Intuitively, the transitions in $M$ are a subset of the one in $T$ since $T$ only removes the transitions that would leave $Q_M$.
For instance, obtaining the meta-controller in Fig.~\ref{fig:meta-system-safe} dictates the removal of transition $((3,C),a,err)$.
Projecting this transition to the environment translates to the removal of transition $(3,a,3)$.
If we compute the $E_d||C||P_{saf}$ for $d$ without transition $(3,a,3)$ results in a system without the transition $((3,C),a,(3,B)$ in addition to the removal of $((3,C),a,err)$.
Formally, in the construction of $T$, only transitions from $Q_T\times Act_F\times Q_F\setminus Q_T$ are removed.
Therefore it follows that $R_M\subseteq R_T$

Combining these two facts, we have that $d = A\setminus del$ computed based on $Q_M$.
The systems $T$ and $M$ have the same set of states and they remove the same transitions that leave $Q_M$ to $Q_F\setminus Q_M$.
Therefore, $d$ could only have more transitions that are related to transitions within $Q_M$, i.e., pictorially, they would be within the shaded area in Fig.~\ref{fig:meta-systems}.
Since the computation of $T$ only deletes transitions that leave the shaded area, it follows that those transitions are also part of $T$.
As a consequence, we have $A\setminus del \subseteq d$.
The direction $d\subseteq A\setminus del$ follows from the fact that $M\subseteq T$.
\end{proof}

\newpage
\subsection*{Case study figures}
\begin{figure}[thpb]
\centering
    \includegraphics[height=60px]{ 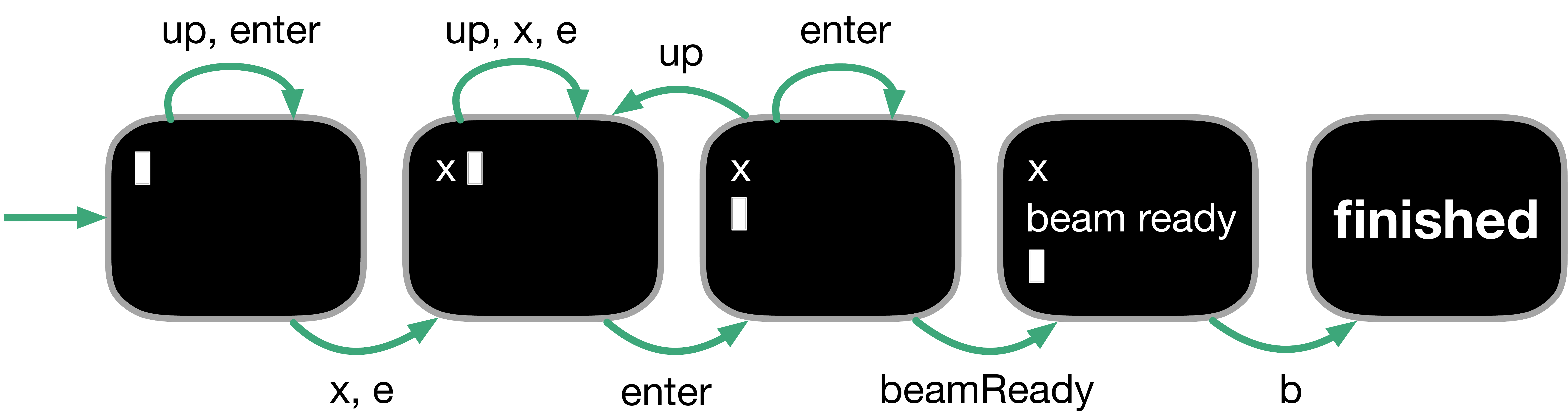}
    \caption{In the Therac-25 case study, $P_{env}$ restricts the environment to firing the beam at most once.}
    \label{fig:p-env-therac}
\end{figure}

\begin{figure}[thpb]
\begin{subfigure}[b]{0.32\textwidth}
    \centering
    \includegraphics[height=90px]{ 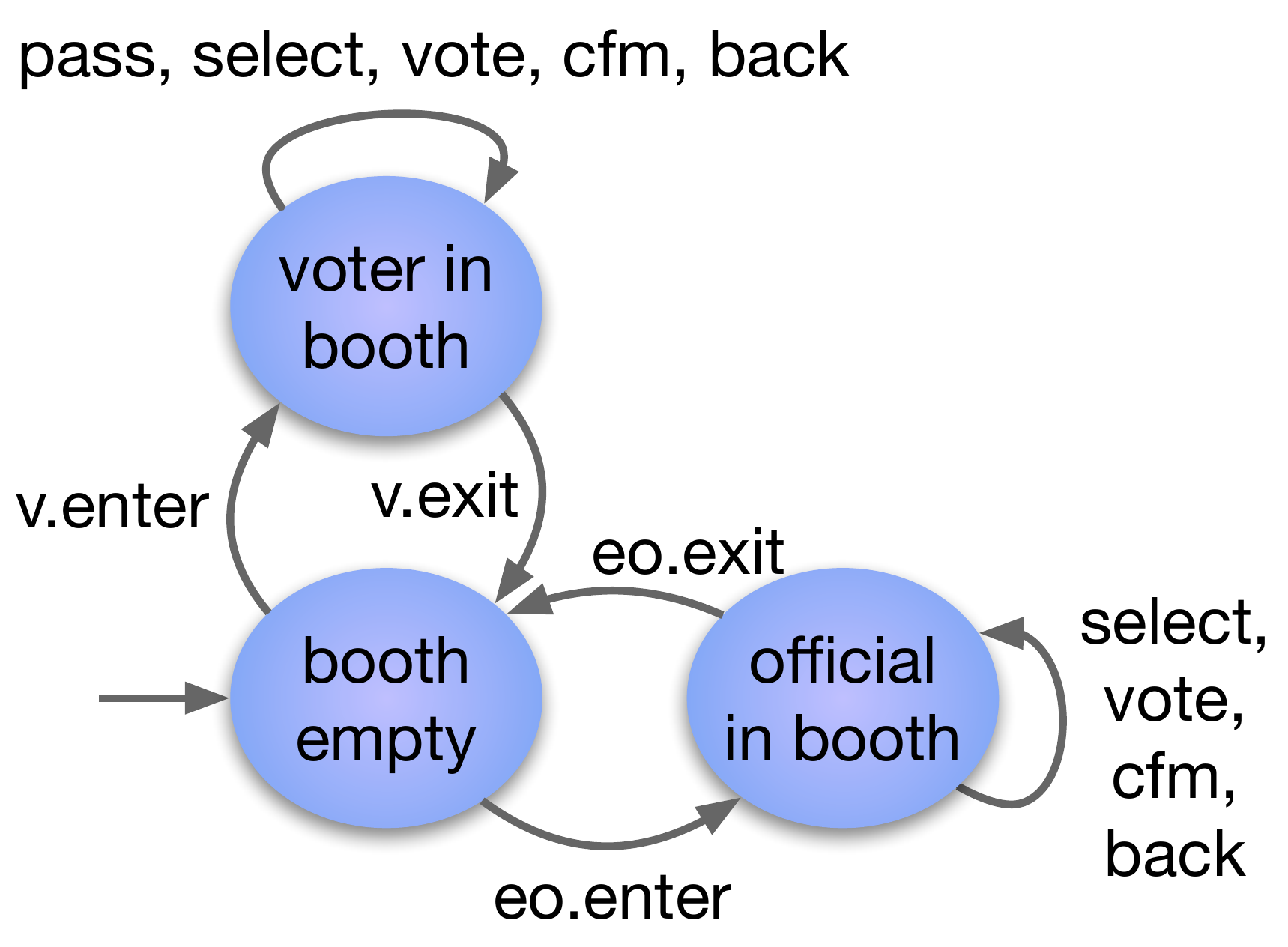}
    \caption{Voting booth model.
    }
    \label{fig:voting-booth}
\end{subfigure}
\enskip
\begin{subfigure}[b]{.65\textwidth}
    \centering
    \includegraphics[height=85px]{ 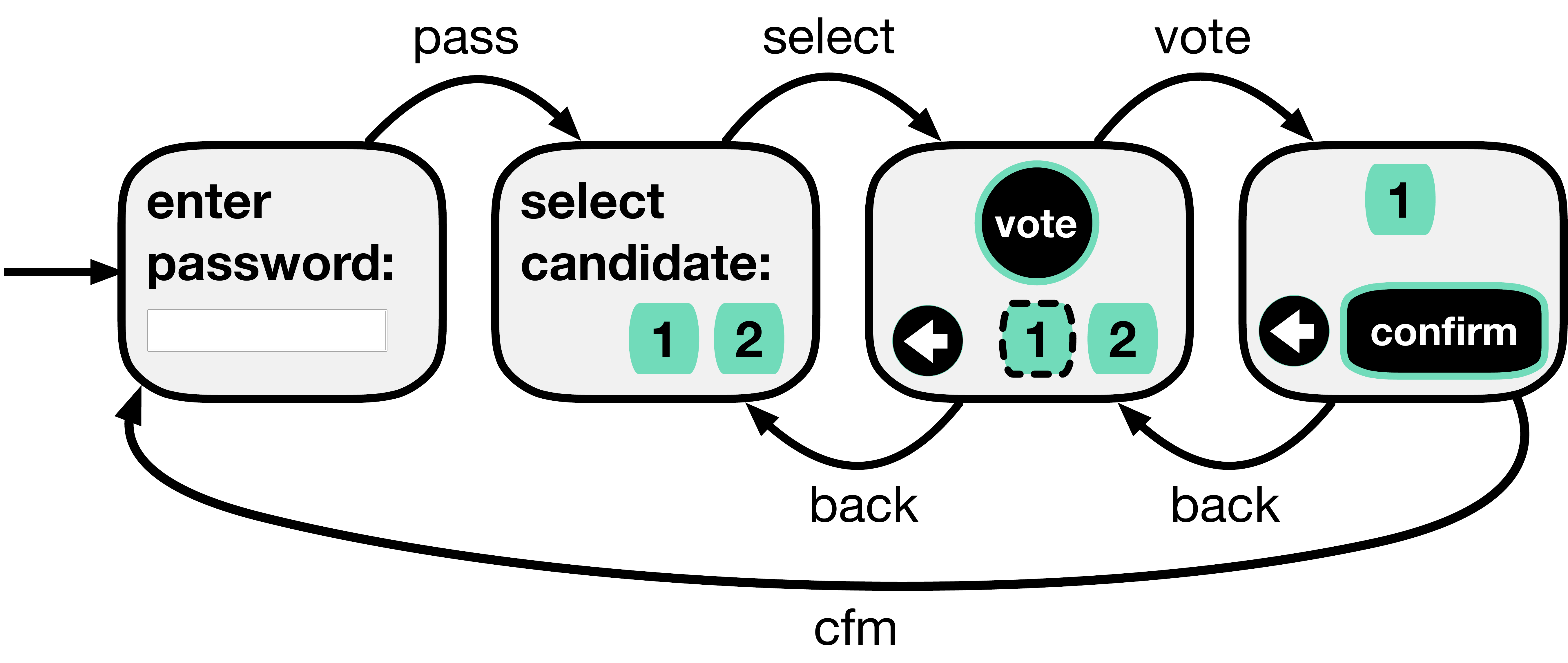}
    \caption{User interface for the voting machine.}
    \label{fig:voting-machine-screen}
\end{subfigure}

\caption{
Models for the voting machine example.
In the figures above, the prefix ``v'' represents actions by the voter, while the prefix ``eo'' represents actions by the election official.}
\end{figure}

\begin{figure}[thpb]
\centering
    \includegraphics[width=350px]{ 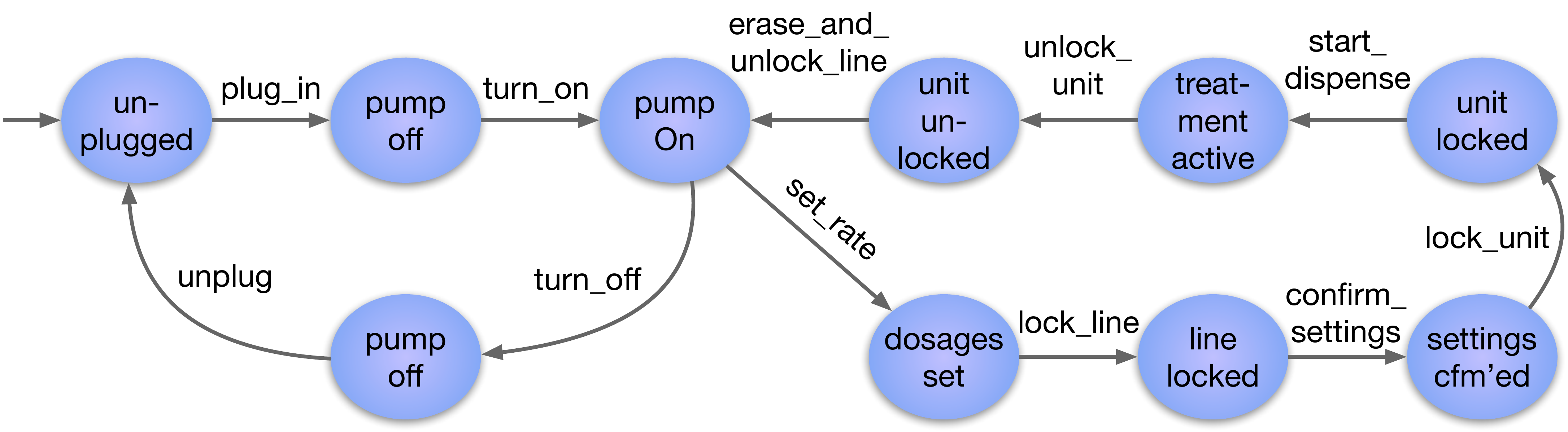}
    \caption{The normative environment for the PCA pump.}
    \label{fig:pump-env}
\end{figure}

\begin{figure}[thpb]
\centering
    \includegraphics[width=350px]{ 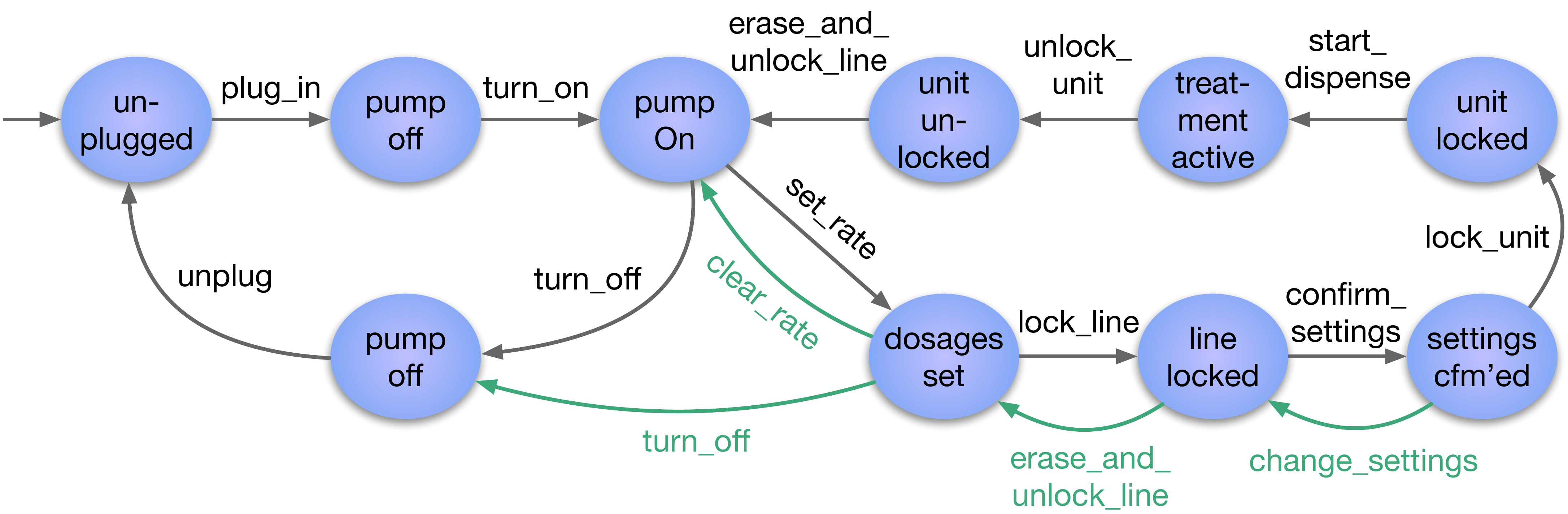}
    \caption{The sole maximal robust deviation for the PCA pump.
    Green arrows indicate transitions that are not in the normative environment.}
    \label{fig:pump-tol}
\end{figure}

\end{document}